\documentclass{IEEEtran}
\usepackage{graphics}
\usepackage{stfloats}
\usepackage{epsfig}
\usepackage{times}
\usepackage{amsthm}
\usepackage{amssymb}
\usepackage{amsmath}
\usepackage{indentfirst}
\usepackage{epstopdf}
\usepackage{multirow}
\usepackage{float}
\usepackage{color}
\usepackage{enumerate}
\usepackage{subfigure}
\usepackage{comment}
\usepackage{cite}
\usepackage{caption}
\usepackage[ruled,vlined]{algorithm2e}
\usepackage{algpseudocode}

\usepackage{url}
\usepackage{pgf}
\usepackage{tikz}
\usetikzlibrary{positioning, arrows.meta,automata}
\usetikzlibrary{calc}
\usepackage[latin1]{inputenc}

\makeatletter
\newif\if@restonecol
\makeatother

\usepackage[ruled,vlined]{algorithm2e}
\usepackage{algpseudocode}

\newcommand{\ba}{\begin{array}}
	\newcommand{\ea}{\end{array}}
\newcommand{\be}{\begin{equation}}
\newcommand{\ee}{\end{equation}}
\newcommand{\bea}{\begin{eqnarray}}
\newcommand{\eea}{\end{eqnarray}}
\newcommand{\bean}{\begin{eqnarray*}}
	\newcommand{\eean}{\end{eqnarray*}}
\newcommand{\bc}{\begin{center}}
	\newcommand{\ec}{\end{center}}

\title{A Uniform Framework for   Diagnosis of Discrete-Event Systems with Unreliable Sensors using Linear Temporal Logic}

\author{Weijie~Dong,~\IEEEmembership{Student Member,~IEEE,}
	Xiang~Yin,~\IEEEmembership{Member,~IEEE,}
	Shaoyuan Li,~\IEEEmembership{Senior Member,~IEEE}
	\thanks{This work was supported by    the National Natural Science Foundation of China (62061136004, 62173226, 61833012) and the National Key Research and Development Program of China (2018AAA0101700).}
	\thanks{W. Dong X. Yin and S. Li  are with Department of Automation and Key Laboratory of System Control and Information Processing,
		Shanghai Jiao Tong University, Shanghai 200240, China.
		E-mail: {\tt\small  \{wjd\_dollar,yinxiang,syli\}@sjtu.edu.cn}. (Corresponding Author: Xiang Yin) } 
}
\newtheorem{mydef}{Definition}
\newtheorem{mythm}{Theorem}

\newtheorem{myexm}{Example}
\newtheorem{remark}{Remark}

\begin{document}

	\maketitle \pagestyle{plain}
	
\begin{abstract}
In this paper, we investigate the diagnosability verification problem of partially-observed discrete-event systems (DES) subject to unreliable sensors. In this setting, upon the occurrence of each event, the sensor reading may be non-deterministic due to measurement noises or possible sensor failures. 
Existing works on this topic mainly consider specific types of unreliable sensors such as the cases of intermittent sensors failures, permanent sensor failures or their combinations. 
In this work, we propose a novel \emph{uniform framework} for diagnosability of DES subject to, not only sensor failures, but also a very general class of unreliable sensors. Our approach is to use linear temporal logic (LTL) with semantics on infinite traces to describe the possible behaviors of the sensors. A new notion of $\varphi$-diagnosability is proposed as the necessary and sufficient condition for the existence of a diagnoser when the behaviors of sensors satisfy the  LTL formula $\varphi$. Effective approach is provided to verify this notion.  We show that, our new notion  of $\varphi$-diagnosability subsumes all existing notions of robust diagnosability of DES subject to sensor failures. Furthermore, the proposed framework is  user-friendly and flexible since it supports an arbitrary user-defined  unreliable sensor type based on the specific scenario of the application. As examples, we provide two new notions of diagnosability, which  have never been investigated in the literature, using our uniform framework. 
\end{abstract}
	
\section{Introduction}	  
\subsection{Backgrounds and Motivations} 
Engineering cyber-physical systems (CPS), such as flexible manufacturing systems, intelligent transportation systems and power systems, are generally very complex due to their intricate operation logic and hybrid dynamics. For such large-scale  systems, \emph{failures} during their operations are very common as millions of their components or sub-systems are working in parallel under uncertain environments. Therefore, \emph{failure diagnosis and detection} are crucial but challenging tasks in order to monitor the operation conditions and to ensure safety for safety-critical CPSs. 
 
In this paper, we investigate the fault diagnosis problem in the framework of discrete-event systems (DES), which is a class of systems widely used in modeling the high-level logical behaviors of CPSs \cite{cassandras2008introduction}.  In the context of DES, the problem of fault diagnosis was initiated by \cite{lin1994diagnosability,sampath1995diagnosability}, where the notion of \emph{diagnosability} was proposed. It is assumed that the event set of the system is partitioned as observable events and unobservable events, and a system is said  to be diagnosable if the occurrence of fault event can always be detected within a finite delay based on the observation sequence. In the past years, fault diagnosis of DES remains a hot topic due to its importance; see, e.g., some recent works \cite{lin2017n,yin2017decidability,basile2017diagnosability,lefebvre2007diagnosis,takai2017generalized,chen2018revised,ran2018codiagnosability,ran2019enforcement,yin2019robust,viana2019codiagnosability,hu2020design,hu2021diagnosability,ma2021marking,cao2021weak,pencole2022diagnosability}. The reader is referred to the comprehensive survey paper and textbook \cite{zaytoon2013overview,lafortune2018history,basilio2021analysis} for more details on this topic.

In the modeling of DES, the occurrences of   events are essentially detected by the corresponding \emph{sensors}. 
In practice, however, due to measurement noises, sensor failures, communication losses or even malicious attacks, 
the sensor readings can be \emph{unreliable} or \emph{non-deterministic} so that the occurrences of  events, which can be detected perfectly in the ideal case, may not always be observed correctly.  For example, the sensor for an observable event may \emph{fail} in the sense that the occurrence of the underlying event cannot be detected. The failures of the sensors can be either intermittent or permanent depending on whether or not the failures can be recovered.  Also, in the networked setting, even when the sensors can always detect their underlying events, their readings still need to be transmitted from the sensors to the diagnoser via communication channels, in which packet losses are possible.  Therefore, the unreliability of sensors are  practical and non-negligible issues in the analysis of diagnosability.

\subsection{Literature Review on Diagnosis with Unreliable Sensors}
In the context of DES, robust fault diagnosis for systems subject to unreliable sensors (or subject to sensor failures) has been studied by many researchers. The reader is referred to  the very recent comprehensive survey papers \cite{boussif2021intermittent,carvalho2021comparative} for more details on this topic.  
For example, in \cite{carvalho2012robust}, the authors investigated the verification of diagnosability for systems subject to \emph{intermittent sensor failures}. In the setting, it is assumed for each observable event, its sensor may fail or recover freely at each time instant. The scenario of intermittent sensor failures can be considered as an instant of \emph{non-deterministic observations} \cite{thorsley2008diagnosability,athanasopoulou2010maximum,takai2012verification}. For example, in \cite{takai2012verification}, the authors proposed to use Mealy automata with non-deterministic output function to capture the scenario where multiple observations are possible for the same transition.

Diagnosability analysis for systems subject to \emph{permanent sensor failures} were studied in \cite{carvalho2013robust,kanagawa2015diagnosability,tomola2017robust,wada2021decentralized}.  In this context, once a sensor fails, it will never recover, i.e., the underlying event becomes unobservable thereafter.  
In \cite{carvalho2013robust}, the authors investigated the robust diagnosability verification problem by assuming that  permanent sensor failures  happen only before the first occurrence of underlying events. This assumption was then relaxed by  \cite{kanagawa2015diagnosability}, which allows that permanent sensor failures can happen at any instant. The results in \cite{carvalho2013robust} and \cite{kanagawa2015diagnosability} have also been extended to the decentralized setting by \cite{tomola2017robust} and \cite{wada2021decentralized}, respectively.

More recently, attempts have been made to unify the notions of diagnosability for systems subject to intermittent sensor failures and permanent sensor failures. For example, it was showed in \cite{carvalho2021comparative} that diagnosability for systems subject to both intermittent and permanent sensor failures can be transferred to the notion of general robust diagnosability under model uncertainty \cite{carvalho2011generalized}.  
In practice, different sensors in the system may belong to different failure types: some may recover after the failure but some may not. 
To address this issue, in \cite{takai2021general}, the author proposed a general framework for robust diagnosability analysis which supports both the case of intermittent  sensor failures and the case of    permanent sensor failures uniformly. 

The effects of unreliable sensors are also closely related to the topic of \emph{networked DES}; see, e.g., \cite{lin2014control,shu2015supervisor,nunes2018codiagnosability,sasi2018detectability,tai2022new}. In this setting,   sensors need to send their readings to the user,  which can be a diagnoser or a controller depend on the context, via communication channels.  Packet delays or losses in the channels then also result in non-deterministic observations, which is very similar to the effect of unreliable sensors. For example, in \cite{oliveira2022k}, the authors studied the verification of $K$-loss diagnosability by assuming that the communication channel can only have a bounded number of consecutive   observation losses. Finally, besides the analysis of diagnosability, the effects of unreliable sensors have also been studied for other  state estimation problems \cite{yin2017initial,tong2021state,zhou2021detectability} as well as  the supervisory control problem \cite{rohloff2005sensor,ushio2016nonblocking,yin2017supervisor,alves2021robust}.   

\subsection{Our Approach and Contributions} 
In the aforementioned existing works on diagnosability analysis for DES subject to unreliable sensors, one  needs to develop customized techniques for different types of sensor failures or sensor modes. For example, in \cite{takai2021general}, where intermittent and permanent sensor failures are unified, one needs to build the sensor automaton that captures all possible switching between the normal mode and the failure mode.  Here, we observe that, each type of sensor failures is actually a \emph{constraint} on how the behaviors of the sensors can change. Such a sensor constraint is essentially a \emph{linear-time property} on the sequence of the sensor modes. For example, the case of permanent sensor failures is actually a safety property requiring that once the sensor goes to a failure mode, it should not recover anymore. 

Motivated by the above observation, in this paper, we propose a  new \emph{uniform framework} for fault diagnosis of DES subject to unreliable sensors. In contrast to existing works, our approach does not rely on any pre-specified type of sensors. Instead, we propose to use Linear Temporal Logic (LTL) formulae as a general tool to describe an arbitrary type of unreliable sensors.  
To this end,  we first adopt the model of  Mealy automata  with non-deterministic output functions proposed in \cite{takai2012verification,ushio2016nonblocking} as the \emph{unconstrained  observation space} of the sensors.  Then we further use an LTL formula $\varphi$, which is referred to as the \emph{sensor constraint} and its specific form is scenario dependent, to describe how sensors can behave within the uncontrained space.   
We propose a new diagnosability condition, called $\varphi$-diagnosability,  following an \emph{assumption based} type of reasoning. 
That is,  whenever the sensors behave following the assumption $\varphi$, the diagnoser should be able to detect the fault within a finite number of steps. We provide an effective procedure for verifying this new condition. In particular, the complexity of the verification procedure is exponential in the number of operators in LTL formula $\varphi$ but is only polynomial in the size of the plant model.  

Our uniform framework is developed based on an arbitrary LTL formula $\varphi$ without its specific meaning. To show the generality of our framework, we further show how the proposed notion of $\varphi$-diagnosability subsumes existing notions of robust diagnosability by explicitly writing down $\varphi$ for different scenarios  of unreliable  sensors. In particular, we show that existing notions of  diagnosability for DES subject to 
(i) intermittent sensor failures  \cite{thorsley2008diagnosability,carvalho2012robust,takai2012verification}; 
(ii) permanent sensor failures  \cite{carvalho2013robust,kanagawa2015diagnosability}; and 
(iii) bounded sensor losses \cite{oliveira2022k}, 
can all be  captured within  our framework in a uniform manner. 

Furthermore, since our notion of  $\varphi$-diagnosability is very generic and LTL is a very expressive and human-language-like approach for presenting linear-time properties, our framework, in fact,  provides a very general and user-friendly way to study diagnosability under arbitrary user-defined or scenario-specified unreliable sensors. As examples, we further introduce two practical scenarios of unreliable sensors, which have never been studied in the literature. 
One is the scenario of \emph{minimum dwell-time} of sensor modes, i.e, whenever a sensor fails, it cannot recover immediately until some pre-specific time period.  
Another is the scenario of \emph{output fairness} in the setting of non-deterministic observations, where each fair observation symbol will eventually be observed if it has infinite chances to be observed.

We note that using linear temporal logic in diagnosability analysis of DES has been investigated in the literature. For example, in \cite{jiang2004failure,chen2015fault}, the authors investigated how to detect faults described by the violations of   LTL formulae. In \cite{cassez2012complexity,bittner2022diagnosability,tuxi2022diagnosability,pencole2022diagnosability}, the verification of standard diagnosability (with reliable sensors) was solved using LTL model checking techniques. However, here, we \emph{do not} aim to detect fault described by LTL or to solve diagnosability verification using LTL model checking.  The role  of LTL in our framework is very different from existing purposes. We use LTL formulae as a general and systematic way to describe different types of unreliable sensors, which, to the best of our knowledge, has never been used in the literature. 

\subsection{Organizations} 
The remaining part of the paper is organized as follows.
Section~\ref{sec:Preliminaries} presents the basic  observation model and the standard notion of diagnosability.  
In Section~\ref{sec:Describing Unreliable Sensors using LTL Formulae}, 
we describe how to use LTL formulae to model unreliable sensors.   
In Section~\ref{sec:Diagnosability under  Sensor Constraints}, we introduce the notion of $\varphi$-diagnosability as the necessary and sufficient condition for the existence of a diagnoser that works correctly under sensor  constraint  $\varphi$.  In Section~\ref{sec:Verification of varphi-diagnosability}, an effective approach is developed for the  verification of $\varphi$-diagnosability. 
In Section~\ref{sec:Application}, we show explicitly how our uniform framework 
subsumes many existing notions as well as supports new types of unreliable sensors that have not been considered.  Finally, we conclude the paper in Section~\ref{sec:Conclusion}. 
Our preliminary idea of using LTL formulae to described sensor behaviors was initially presented in \cite{dong2022fault}. However, \cite{dong2022fault} only focuses on a very specific type of non-deterministic sensors, i.e., output fairness. The framework in the present paper is much more general, where the result in \cite{dong2022fault} is a now very special instance as described in Section~\ref{subsec:fair}.

\section{Preliminaries}\label{sec:Preliminaries}

\subsection{Partially-Observed Discrete Event Systems}
We assume basic knoweledge in formal languages and discrete-event systems  \cite{cassandras2008introduction}. 
Let $\Sigma$ be a set of events or alphabets.  
A string $s=\sigma_1\sigma_2\cdots \sigma_n (\cdots)$ is a finite (or infinite) sequence of events. 
For a finite string $s$, we denote by $|s|$ the length of $s$, which is the number of events in it.  
We denote by $\Sigma^*$ and $\Sigma^{\omega}$ the set of all finite and infinite strings over $\Sigma$, respectively.  
We write $\Sigma^+=\Sigma^*\cup \Sigma^\omega$ as the set of finite or infinite strings. 
Note that $\Sigma^*$ includes the empty string $\varepsilon$. 
A $*$-language $L\subseteq \Sigma^*$ is a set of finite strings and an $\omega$-language $L\subseteq \Sigma^\omega$ is a set of infinite strings. 
For any $*$-language $L\subseteq \Sigma^*$, its \emph{prefix-closure} is a  $*$-language 
$\overline{L}=\{s\in \Sigma^*: \exists t \in \Sigma^* \text{ s.t. }  s t\in L\}$. 
Analogously, for any $\omega$-language $L \subseteq \Sigma^\omega$,  its prefix-closure is defined as the set of all its finite prefixes, which is $*$-language $\overline{L}=\{s \in \Sigma^*:\exists t \in \Sigma^\omega \text{ s.t. } s t \in L\}$. 
%
%
Given an infinite string $s \in \Sigma^\omega$,  $\textsf{Inf}(s)$ denotes the set of events that appear infinite number of times in $s$.

In this work, we consider a DES modeled by a deterministic finite-state automaton 
\[
G=(Q,\Sigma,\delta,q_0),
\]
where  
$Q$ is the finite set of states, 
$\Sigma$ is the finite set of events,
$\delta:Q \times \Sigma \rightarrow Q$ is the partial deterministic transition function such that: 
for any $q,q'\in Q,\sigma\in \Sigma$, $\delta(q,\sigma)=q'$ means that there exists a transition from state $q$ to state $q'$ labeled by event $\sigma$, and
$q_0 \in Q$ is the initial state. 
Note that, we do not consider marked states in the system model since it is irrelevant to our purpose of diagnosability verification. 
The transition function $\delta$ is also extended to  
$\delta: Q\times \Sigma^*\to Q$ recursively by: 
for any $q\in Q,s\in \Sigma^*$ and $\sigma\in \Sigma$, we have
$\delta(q,\varepsilon)=q$  and  $\delta(q,s\sigma)=\delta(\delta(q,s),\sigma)$. 
Then the $*$-language generated by $G$ is defined by 
$\mathcal{L}(G)=\{s\in \Sigma^*:\delta(q_0,s)!\}$, where ``!" means ``is defined", 
and the $\omega$-language generated by $G$ is defined by 
$\mathcal{L}^\omega(G)=\{s\in \Sigma^\omega: \forall s'\in \overline{\{s\}}, \delta(q_0,s')!\}$. 
We assume that system $G$ is live, i.e., for any state $q\in Q$, there exists an event $\sigma\in \Sigma$ such that 
$\delta(q,\sigma)!$.  

In the partial observation setting, the occurrence of each event may not be perfectly observed. 
The limited sensor capability is usually modeled by an observation  mask
\begin{equation}\label{eq:mask}
	O: \Sigma\to \Delta \dot{\cup} \{\varepsilon\}, 
\end{equation} 
where $\Delta$ is a new set of observation symbols.  
That is, we observe $O(\sigma)$ upon the occurrence of event $\sigma\in \Sigma$.
We say event $\sigma\in \Sigma$ is observable if $O(\sigma)\in \Delta$ and unobservable if $O(\sigma)=\varepsilon$.
The observation mask is also extended to $O:\Sigma^+\to \Delta^+$ by: 
for any $s\in \Sigma^+$, $O(s)$ is obtained by replacing each event $\sigma$ in string $s$ as $O(\sigma)$. 

\subsection{Observations under Unreliable Sensors}
The observation mask as defined in Equation~\eqref{eq:mask} corresponds to the case of \emph{reliable sensors}. That is, if event $\sigma\in \Sigma$ is observable, then the sensor reading will always be $O(\sigma)$ whenever it occurs. 
In many real-world scenarios, however, sensors may be \emph{unreliable} due to multiple reasons. 
First, due to possible  \emph{sensor failures}, an observable event may become unobservable either permanently or intermittently.  
Second, due to measurement noises, the sensor reading may be \emph{non-deterministic} in the sense that we may observe different symbols for difference occurrences of the same event.  
Furthermore, the sensor reading may be \emph{state-dependent} in the sense that the same event may have different readings at different states.  

To capture the  observations under unreliable sensors, in  \cite{takai2012verification,ushio2016nonblocking}, the authors proposed to use the following state-dependent (or transition-based) non-deterministic observation mapping 
\begin{equation} \label{eq:non-deterministic observation mapping}
	\mathcal{O}: Q\times \Sigma \to 2^{\Delta\cup\{\varepsilon\}}.
\end{equation}
Specifically,  $\Delta$ is the set of all possible observation symbols.  
For any state $q\in Q$ and event $\sigma\in \Sigma$, 
$\mathcal{O}(q,\sigma)$ denotes \emph{the set of all possible observations} if event $\sigma$ occurs at state $q$, 
i.e., the sensor reading may be any symbol $o\in \Delta$ or $\varepsilon$ in $\mathcal{O}(q,\sigma)$  non-deterministically. 

The observation mapping $\mathcal{O}$ can also be extended to 
$\mathcal{O}: Q \times \Sigma^+ \to 2^{\Delta^+}$ by: 
for any state $q \in Q$ and string $s=\sigma_1\sigma_2\cdots \in \Sigma^+$, 
we have $o_1 o_2 \cdots\in \mathcal{O}(q,s)$ if 
\begin{equation}
	\forall i=1,2,\dots: o_i\in \mathcal{O}( f(q,\sigma_1\dots \sigma_{i-1}),\sigma_i ),
\end{equation} where $\sigma_{0}:=\varepsilon$.
Since the observations are non-deterministic, in general,  $\mathcal{O}(q,s) \subseteq \Delta^+$ is not a singleton and each sequence in $\mathcal{O}(q,s)$ is called an  \emph{observation realization} of string $s$. 
For the sake of brevity, we use $\mathcal{O}(s)$ to denote $\mathcal{O}(q,s)$  if the state $q$ is the initial state, i.e., $q=q_0$.

To incorporate the actual observation into the internal execution of the system,  it is convenient to define 
\[
\Sigma_e=Q\times \Sigma\times (\Delta \cup \{\varepsilon\})
\]
as the set of \emph{extended events}. 
Then, an extended string is a finite or infinite sequence of extended events.
We say an extended string
$s=(q_0,\sigma_0,o_0)(q_1,\sigma_1,o_1)\dots  \in \Sigma_e^+$ is generated by $G$ if 
for any $i=1,2,\dots$, we have $f(q_i,\sigma_i)=q_{i+1}$ and $o_i\in \mathcal{O}(q_i,\sigma_i)$. 
We denote by $\mathcal{L}_e(G)$ and $\mathcal{L}_e^\omega(G)$ the set of all finite and infinite extended strings generated by $G$, respectively.  
Essentially, an extended string contains three parts of information: 
(i)  the state sequence visited, 
(ii)  the event sequence generated, and 
(iii) the output sequence observed. Then  for any extended string
$s=(q_0,\sigma_0,o_0) (q_1,\sigma_1,o_1)\cdots \in \Sigma_e^+$, 
we define 
$\Theta_{Q}(s)=q_0 q_1\cdots\in Q^+$, 
$\Theta_{\Sigma}(s)=\sigma_0\sigma_1 \cdots\in \Sigma^+$. and $\Theta_{\Delta}(s)=o_0 o_1\cdots\in (\Delta\cup \{\varepsilon\})^+$ 
as its corresponding state sequence, (internal) event string and output string, respectively.
Clearly, for any  
$s \in \Sigma_e^+$, we have 
$\Theta_{\Delta}(s) \in \mathcal{O}(\Theta_{\Sigma}(s))$
because $\Theta_{\Delta}(s)$ is a specific observation realization of internal string $\Theta_{\Sigma}(s)$.  

\begin{myexm} \label{exam:system example}\upshape
	Let us consider system $G_1$ shown in Figure~\ref{fig:System G_1}, where  $\Sigma=\{a,b,c,f,u\}$ and $\Delta=\Sigma$. 
	The output function $\mathcal{O}: Q\times \Sigma \!\to\! 2^{\Delta\cup\{\varepsilon\}}$ is specified by the label of each transition, 
	where the LHS of   ``$\backslash $" denotes the internal event 
	and  the RHS of   ``$\backslash $" denotes the set of all possible output symbols. 
	For example, $b \backslash  \{b,\varepsilon\}$ from state $3$ to state $4$ means that 
	at state $3$, upon the occurrence of event $b$, the system moves to state $4$, i.e., $\delta(3,b)=4$,  and we may observe either $b$ or nothing, i.e.,  $\mathcal{O}(3,b)= \{b,\varepsilon\}$. 
	Then for finite string  $f a b \in \mathcal{L}(G_1)$, 
	the set of all possible observation realizations is $\mathcal{O}(f a b)=\{a b,a\}$. 
	These two different observations lead to the following two extended strings   $(1, f,\varepsilon)(2,a,a)(3,b,b)\in \mathcal{L}_e(G)$ and  $(1,f ,\varepsilon)(2,a,a)(3,b,\varepsilon)\in \mathcal{L}_e(G)$, respectively.
\end{myexm}
\begin{figure} 
	\centering
	\begin{tikzpicture}[->,>={Latex}, thick, initial text={}, node distance=1.25cm, initial where=above, thick, base node/.style={circle, draw, minimum size=5mm, font=\small}]		
		\node[initial, state, base node, ] (1) at (1.5,7) {$1$}; 
		\node[state, base node, ](2) at (0.5,6) {$2$};
		\node[state, base node, ](3) [below of=2] {$3$};
		\node[state, base node, ](4) [below of=3] {$4$};
		\node[state, base node, ](5) at (2.5,6) {$5$};
		\node[state, base node, ](6) [below of=5] {$6$};
		\node[state, base node, ](7) [below of=6] {$7$};
		\node[state, base node, ](8) [below of=7] {$8$};
		
		\path[]
		(1) edge node [left, xshift=0.1cm, yshift=0.15cm] {\fontsize{8}{1} $f \backslash\! \{\varepsilon\} $} (2)
		(1) edge node [right, xshift=-0.2cm, yshift=0.15cm] {\fontsize{8}{1} $u \backslash\! \{\varepsilon\}$} (5)
		(2) edge node [right, xshift=-0.1cm] {\fontsize{8}{1} $a\! \backslash\! \{a\} $} (3)
		(3) edge node [midway,right, xshift=-0.1cm] {\fontsize{8}{1} $b\! \backslash\! \{b,\varepsilon\} $} (4);
		
		\draw (4.west) -- ($(4.west)+(-0.75,0)$)
		($(4.west)+(-0.75,0)$) -- node [right, xshift=-0.15cm] {\fontsize{8}{1} $c\! \backslash\! \{c\} $} 
		($(2.west)+(-0.75,0)$)
		-- (2.west);
		
		\path[]
		(5) edge node [midway,right, xshift=-0.1cm] {\fontsize{8}{1} $a  \backslash\! \{a\} $} (6)
		(6) edge node [midway,right, xshift=-0.1cm] {\fontsize{8}{1} $b  \backslash\! \{b,\varepsilon\}$} (7)
		(7) edge node [right, xshift=-0.1cm ] {\fontsize{8}{1} $ c  \backslash\! \{c\} $} (8);
		
		\draw
		(8.west) -- ($(8.west)+(-0.55,0)$)
		($(8.west)+(-0.55,0)$) -- node [midway,right, xshift=-0.15cm] {\fontsize{8}{1} $a\! \backslash\! \{a\} $} ($(7.west)+(-0.55,0)$)
		-- (7.west);
		
	\end{tikzpicture}
		\caption{\label{fig:System G_1}System $G_1$.}
	\end{figure}
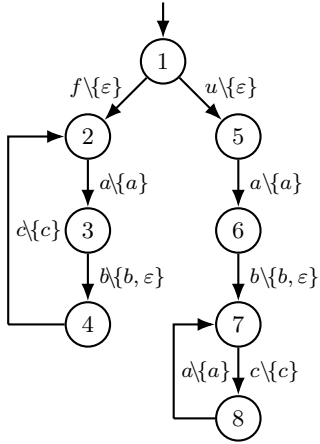
	
	\begin{remark} \label{rem:capture intermittent}\upshape
		The non-deterministic observation mapping $Q\times \Sigma \to 2^{\Delta\cup\{\varepsilon\}}$ is quite general and subsumes many observation models in the literature. 
		For example, for the standard natural projection, we assume that $\Sigma_o\subseteq \Sigma$ is the set of observable events and let $\Delta=\Sigma_o$.  
		Then one can define $\mathcal{O}$ by: 
		$\mathcal{O}(q,\sigma)=\{\sigma\}$ for all $\sigma\in \Sigma_{o}$
		and 
		$\mathcal{O}(q,\sigma)=\{\varepsilon\}$ for all $\sigma\notin \Sigma_{o}$. 
		Also, it captures the so-called \emph{intermittent loss of observations} \cite{carvalho2012robust,lin2014control}.  In this setting,  the event set is usually partitioned as 	$\Sigma=\Sigma_{r} \dot{\cup} \Sigma_{ur} \dot{\cup} \Sigma_{uo}$, 
		where $\Sigma_{r}$ is the set of reliable events whose occurrences can  always be observed directly, $\Sigma_{ur}$ is the set of unreliable events whose occurrences may be observed but can also be lost and  $\Sigma_{uo}$ is the set of unobservable events whose occurrences can  never be observed. 
		This setting can be captured by considering  $\mathcal{O}$ such that  $\Delta=\Sigma_r\cup \Sigma_{ur}$ and  for any $q \in Q$ and  $\sigma \in \Sigma$, we have
		\begin{align}\label{eq:intermittent observation}
			\mathcal{O}(q,\sigma)\!=\!
			\left\{\! 
			\begin{array}{l l}
				\{\sigma\} &\text{if } \sigma \!\in\! \Sigma_r \\
				\{\sigma,\varepsilon\} &\text{if } \sigma \!\in\! \Sigma_{ur} \\ 
				\{\varepsilon\} &\text{if } \sigma \!\in\! \Sigma_{uo} 
			\end{array} 
			\right..
		\end{align}
		Note that, for the   general case  we consider here, the output symbols $\Delta$ can be different from the original event set $\Sigma$. 
	\end{remark}

\subsection{Fault Diagnosis}
In the context of fault diagnosis of DES, it is assumed that the system is subject to faults, which are modeled by a set of fault events $\Sigma_F \subset \Sigma$. 
For the sake of simplicity, we do not distinguish among different fault types in this work. 
We say a string $s\in \Sigma^+$ is faulty if it contains a fault event in $\Sigma_F$ and we write $\Sigma_F\in s$ with a slight abuse of notation; otherwise, we call string $s$ a  normal string. 
We denote by  $\mathcal{L}_F(G)$ and $\mathcal{L}_F^\omega(G)$ as the sets of all finite and infinite faulty strings generated by $G$, respectively. 
Similarly, we define $\Sigma_{e,F}= Q\times \Sigma_F\times (\Delta\cup \{\varepsilon\}) \subset \Sigma_e$ as the set of extended fault events 
and we denote by    $\mathcal{L}_{e,F}(G)$ and $\mathcal{L}_{e,F}^\omega(G)$ as the sets of all finite and infinite extended faulty strings generated by $G$, respectively.
Finally, we define $\Psi_e(G)$ as the set of all finite extended faulty strings in which extended fault events occur \emph{for the first time}, i.e., 
\begin{align}
	\Psi_e(G)=\{s \in \mathcal{L}_{e,F}(G): \forall t \in \overline{\{s\}} \setminus \{s\},  \Sigma_{e,F} \notin  t \}. \nonumber
\end{align} 

To capture whether or not the occurrences of fault events can be always detected within a finite number of steps, the notion of diagnosability (under non-deterministic observations) has been proposed in the literature \cite{takai2012verification}. 

\begin{mydef}[Diangnosability] \label{def:classical diag}\upshape
	System $G$ is said to be \emph{diagnosable} w.r.t.\    mapping $\mathcal{O}$ and fault events $\Sigma_F$ if 
	\begin{equation} \label{eq:classical diagnosis}
		(\forall s \!\in\! \Psi_e(G) ) (\exists n\!\in\! \mathbb{N}) (\forall st  \!\in\! \mathcal{L}_e(G)) 
		[ |t| \geq n \Rightarrow \textsf{diag} ], 
	\end{equation}
	where the diagnostic condition $\textsf{diag}$ is 
	\[
	(\forall \omega \in \mathcal{L}_e(G))  
	[\Theta_{\Delta}(\omega)=\Theta_{\Delta}(st) \Rightarrow    \Sigma_{e,F}\in \omega]. \vspace{8pt}
	\] 
\end{mydef}

Intuitively, the above definition says that,  
for any faulty extended string in which fault events appear for the first time, there exists a finite detection bound such that,   
for any of its continuation longer than the detection bound, 
any other extended strings having the same observation must also contain fault events, which means we can claim for sure that fault events have occurred.  Note that we consider extended strings $\mathcal{L}_e(G)$ rather than the internal strings $\mathcal{L}(G)$ in order to capture the issue of non-deterministic observations. 

\begin{myexm} \label{exam:classical not diagnosable}\upshape
	Again, we consider system $G_1$ depicted in Figure~\ref{fig:System G_1} and we assume $\Sigma_F=\{f\}$.  
	Let us consider faulty extended string  $(1,f,\varepsilon) \in \Psi_e(G)$, which can be extended  arbitrarily long as
	\[
	s_F=(1,f, \varepsilon)[(2,a,a)(3,b,\varepsilon)(4,c,c)]^n.
	\]
	However, for any $n$, we can find a normal extended string 
	\[
	s_N=(1,u, \varepsilon)(5,a,a)(6,b,\varepsilon)[(7,c,c)(8,a,a)]^{n-1}(7,c,c)
	\]
	such that $\Theta_{\Delta}(s_F)=\Theta_{\Delta}(s_N)=(a c)^n$. 
	Therefore, we know that system $G_1$ is not diagnosable with respect to $\mathcal{O}$ and $\Sigma_F$.
\end{myexm}

\section{LTL-Based Description of Unreliable Sensors}\label{sec:Describing Unreliable Sensors using LTL Formulae}
%

\subsection{Motivation and Main Idea}
In the previous section, we have introduced non-deterministic observation mapping $\mathcal{O}: Q\times \Sigma \to 2^{\Delta\cup\{\varepsilon\}}$ to capture all possible observations. However, this model essentially provides the \emph{possible observation space} for the purpose of worst-case analysis. When the observation non-determinism has a concrete physical meaning, the extended language $\mathcal{L}_e(G)$ may contain  observation realizations that are not feasible in practice.  

To motivate our developments, we take the scenario of \emph{permanent sensor failures} as an example. As we have discussed in Remark~\ref{rem:capture intermittent}, the non-deterministic observation mapping specified in Equation~\eqref{eq:intermittent observation} can capture the scenario of intermittent sensor failures, i.e., for each event $\sigma\in \Sigma_{ur}$, we may observe $\sigma$ when the sensor is normal or $\varepsilon$ when the sensor fails.   
However, in the context of permanent sensor failures, once an unreliable sensor fails, it will not be normal forever.  
Then for system $G_1$  in Figure~\ref{fig:System G_1}, in the context of permanent sensor failures, the following extended string in $\mathcal{L}_e(G)$ will no longer be feasible,  
\[
s_F=(1,f,\varepsilon)(2,a,a)(3,b,\varepsilon)(4,c,c)(2,a,a)(3,b,b) 
\]
because the occurrence of extended event $(3,b,\varepsilon)$ means that the sensor corresponding to event $b$ has already failed and it is not possible to have extended event $(3,b,b)$ thereafter. 

Such a simple scenario of permanent sensor failures cannot be modeled by non-deterministic mapping in the form of Equation~\eqref{eq:non-deterministic observation mapping} directly.  To describe this scenario, different approaches have been developed in the literature \cite{carvalho2013robust,kanagawa2015diagnosability,wada2021decentralized} and the basic idea is to use additional states to capture the failure/normal status of  unreliable sensors.  

One of the motivations of our work is to provide a general framework that unifies the scenarios of intermittent and permanent sensors failures. However, our approach goes much beyond this basic objective and supports arbitrary \emph{user-specified} sensor behaviors. 
Specifically, the basic idea of our approach is as follows:
\begin{itemize}
	\item 
	First, we use non-deterministic observation mapping $\mathcal{O}: Q\times \Sigma \to 2^{\Delta\cup\{\varepsilon\}}$ to generate  the \emph{unconstrained} observation space $\mathcal{L}_e(G)$ without considering the specific setting of each sensor; 
	\item 
	Then, by incorporating how each sensor should behave, which is referred to as the \emph{sensor  constraint} throughout the paper, we further restrict the unconstrained observation space by eliminating those observations that are not feasible in the concrete setting.    
\end{itemize}
To realize the above idea, the key question is how to describe the behaviors of unreliable sensors for  different contexts. Our novelty here is that the user does not need to hand-code the feasible behaviors of the sensors. Instead, we provide a very generic and user-friendly way to describe the sensor  constraints using Linear Temporal Logic (LTL) formulae, which can further be proceeded algorithmically.

\subsection{Linear Temporal Logic}
Let $\mathcal{AP}$ be a finite set of atomic propositions. An LTL formula $\varphi$ is constructed based on a set of atomic propositions $\mathcal{AP}$, Boolean operators and temporal operators as follows:
\[
\varphi::=\texttt{true} \mid  p \mid   \neg \varphi \mid \varphi_1 \wedge \varphi_2 \mid  \bigcirc \varphi \mid   \varphi_1  U \varphi_2 ,
\]
where $p\in \mathcal{AP}$ is an atomic proposition, 
$\neg$ and $\wedge$ stand  for  logical  negation  and  disjunction, respectively, 
while $\bigcirc$ and $U$ denote ``next" and ``until", respectively. 
Note that other Boolean operators can be induced by $\wedge$ and $\neg$, e.g., 
$\varphi_1 \vee \varphi_2 = \neg (\neg \varphi_1 \wedge \neg \varphi_2)$ and $\varphi_1 \rightarrow \varphi_2 = \neg \varphi_1 \vee \varphi_2$. 
Furthermore, it is convenient to define temporal operators 
$\lozenge$  ``eventually" by $\lozenge \varphi = \texttt{true} U \varphi$ and  
$\Box$  ``always"  by $ \Box \varphi = \neg \lozenge \neg \varphi$. 

LTL formulae are evaluated over  infinite  strings of atomic proposition sets, which are also referred to as infinite \emph{words}. 
For any infinite word $s \in (2^\mathcal{AP})^\omega$, we denote by $s \models \varphi$ if it satisfies LTL formula $\varphi$ and use $\textsf{word}(\varphi)$ to denote all strings satisfying $\varphi$, i.e., $\textsf{word}(\varphi)=\{s\in (2^\mathcal{AP})^\omega: s \models \varphi\}$. 
The reader is referred to \cite{baier2008principles}  for more details on the semantics of LTL. 
For example, 
let $s=A_0A_1\dots \in (2^\mathcal{AP})^\omega$ be an infinite word and  $\phi$ be a Boolean formula without temporal operators.  
Then 
$s \models \Box \lozenge \phi$ means that  proposition $\phi$ holds \emph{infinitely often} in $s$,  i.e., $\forall i\geq 0,\exists j>i: \phi(A_j)=\texttt{true}$, while 
$s \models  \lozenge \Box \phi$ means that proposition $\phi$ holds \emph{forever after some finite delays} in $s$, i.e., $\exists i\geq 0,\forall j>i: \phi(A_j)=\texttt{true}$.

Given an LTL formula $\varphi$, to capture all infinite words satisfying $\varphi$, we introduce the notion of Non-deterministic B\"{u}chi Automaton (NBA) as follows.
\begin{mydef}[Non-deterministic B\"{u}chi Automaton]\upshape
	An NBA is a 5-tuple $\mathcal{B}=(X, X_0,\Sigma_B,\xi, X_m)$, where 
	$X$ is a finite set of states, 
	$X_0 \subseteq X$ is the set of initial states, 
	$\Sigma_B$ is an alphabet, 
	$\xi: X \times \Sigma_B \rightarrow 2^{X}$ is a non-deterministic transition function and $X_m \subseteq X$ is the set of accepting states. 
\end{mydef}

Given an infinite word $s= \sigma_0 \sigma_1 \cdots \in \Sigma_B^\omega$, an infinite path of $\mathcal{B}$ induced by $s$ is an infinite state sequence $\rho=x_0 x_1 \cdots \in X^\omega$  such that $x_0\in X_{0}$ and $x_{i+1} \in \xi(x_i,\sigma_i)$ for any $i=0,1, \cdots$.  
An infinite path $\rho$ is said to be \emph{accepted} by NBA $\mathcal{B}$ if it visits accepting states $X_m$ infinitely often, i.e., $\textsf{inf}(\rho) \cap X_m \neq \emptyset$. 
We say  an infinite word $s$ is accepted by $\mathcal{B}$ if it induces an infinite path accepted by $\mathcal{B}$. We denote by $\mathcal{L}^\omega_m (\mathcal{B})$ the set of strings accepted by NBA $\mathcal{B}$. 
The non-deterministic transition function $\xi$ can also be extended to $\xi: X \times \Sigma_B^* \rightarrow 2^{X}$ recursively in the usual manner. 
Also, we use notation $\mathcal{L}(\mathcal{B}) \subseteq \Sigma_B^*$ to denote the set of all finite strings generated by $\mathcal{B}$, i.e., 
$\mathcal{L}(\mathcal{B})=\{  s\in \Sigma_B^*: \exists x_0\in X_0\text{ s.t. } \xi(x_0,s)\neq \emptyset \}$. 

It is well-known that \cite{baier2008principles}, for any LTL formula $\varphi$, we can translate $\varphi$ to an NBA $\mathcal{B}$ over event set $\Sigma_B=2^\mathcal{AP}$ such that $\mathcal{L}^\omega_m(\mathcal{B})=\textsf{word}(\varphi)$ and we say such an  NBA $\mathcal{B}$ is associated with $\varphi$. There are efficient tools in the literature to translate an LTL formula to an NBA such as  \texttt{LTL2BA}  \cite{giannakopoulou2002states}. 

\subsection{Sensor Constraints as LTL Formulae}
To capture the sensor  constraints using LTL formulae, we define a labeling function 
\begin{equation}
	\textsf{label}: \Sigma_e \to 2^\mathcal{AP}
\end{equation} 
that assigns each extended event a set of atomic propositions. 
For any  infinite extended string $s=\sigma_1 \sigma_2\cdots   \in \Sigma_e^\omega$, its \emph{trace} is defined by 
$\textsf{trace}(s)=\textsf{label}(\sigma_1) \textsf{label}(\sigma_2) \cdots  \in (2^\mathcal{AP})^\omega$.  
Let $\varphi$ be an LTL formula specifying the sensor constraint. 
We say an infinite extended string $s\in \Sigma_e^\omega$ is $\varphi$-\emph{compatible}, if $\textsf{trace}(s) \models \varphi$.  
For the sake of simplicity, with   a slight abuse of notation, we also write $s \models \varphi$ whenever $\textsf{trace}(s) \models \varphi$. 
We define  
\[
\mathcal{L}_{e}^{\varphi}(G)=\{s \in \mathcal{L}^{\omega}_{e}(G): s \models  \varphi \} 
\]
as the set of  all $\varphi$-compatible infinite extended strings  generated by $G$. 

Intuitively, the above definition says that if we assume that the behaviors of the sensors satisfy   LTL formula $\varphi$, then only infinite observations in $\mathcal{L}_{e}^{\varphi}(G)$ are feasible in practice.  
Note that, we use infinite observation in order to match the semantic of LTL. 
For the purpose of online diagnosis, we only observe finite observations in   
$\overline{\mathcal{L}^{\varphi}_{e}(G)}$.  
In otehr words, it is impossible to observe finite extended strings in 
$\mathcal{L}_e(G) \setminus (\overline{\mathcal{L}^{\varphi}_{e}(G)}$ since they cannot be extended to infinite strings in $\overline{\mathcal{L}^{\varphi}_{e}(G)}$.
Finally, we also define $\mathcal{L}^{\varphi}_{e,F}(G)=\mathcal{L}^{\varphi}_{e}(G) \cap \mathcal{L}_{e,F}^\omega(G)$ as the set of infinite extended faulty strings that are $\varphi$-compatible. 

Note that, how to select atomic propositions $\mathcal{AP}$, labeling function $\textsf{label}: \Sigma_e \to 2^\mathcal{AP}$ and  LTL formula $\varphi$ is application dependent. In Section~\ref{sec:Application}, we will elaborate on how they are selected for different practical scenarios. 
Here, we consider the scenario of permanent sensor failures to illustrate the choices of 
$\mathcal{AP}$,   $\textsf{label}$ and  $\varphi$, and use this scenario as a running example for the latter developments.

\begin{myexm}[Sensor Constraints for Permanent Sensor Failures]\upshape \label{exam:permanent-specification}
	We still consider  system $G_1$ shown in Figure~\ref{fig:System G_1}. 
	However, now we further assume that the reason why event $b$ may become unobservable is due to a \emph{permanent sensor failure}. This scenario can be described as follows. 
	First, we choose the set of atomic propositions as
	\[
	\mathcal{AP}=\{m_0,m_1\}, 
	\]
	where $m_0$ and $m_1$ represent that the sensor for event $b$ is normal and faulty, respectively. 
	Then we define the labeling function $\textsf{label}: \Sigma_e \to 2^\mathcal{AP}$  by:
	for any $\sigma_e=(q,\sigma,o)\in \Sigma_e$, we have 
	\begin{equation}\label{eq:per-example}
		\textsf{label}(\sigma_e)= \left\{
		\begin{array}{cl}
			\{m_0\},     & \text{if }  \sigma=b\wedge o=b\\
			\{m_1\},   & \text{if }  \sigma=b\wedge o=\varepsilon\\
			\emptyset,  &  \text{otherwise}   	
		\end{array} \right..
	\end{equation}
	Then the sensor constraint for the scenario of permanent failures can be written as 
	\begin{equation} \label{eq:permanent equation}
		\varphi_{\textsf{per}}=\Box( m_1 \to \Box \neg m_0).
	\end{equation}
	Intuitively,  formula $\varphi_{\textsf{per}}$ says that whenever atomic proposition $m_1$ holds, which means that sensor for event $b$ fails, atomic proposition $m_0$, which means that sensor reads the occurrence of $b$ normally, cannot hold anymore.  
	
	Then under the above described sensor constraint $\varphi_{\textsf{per}}$, we know that the following extended string is not $\varphi_{\textsf{per}}$-compatible 
	\[
	s_F=(1,f,\varepsilon)(2,a,a)(3,b,\varepsilon)((4,c,c)(2,a,a)(3,b,b) )^\omega 
	\]
	because
	$\textsf{trace}(s_F)=  \emptyset \emptyset \{m_1\} ( \emptyset \emptyset \{m_0\}    )^\omega  \not\models \varphi_{\textsf{per}}$. 
	Therefore, with sensor constraint $\varphi_{\textsf{per}}$, it is impossible to observe finite sequence
	$a c a b$ because 
	\[
	(1,f,\varepsilon)(2,a,a)(3,b,\varepsilon)(4,c,c)(2,a,a)(3,b,b)
	\!\notin\!  \overline{\mathcal{L}^{\varphi}_{e}(G)}. 
	\]
\end{myexm} 

\section{Diagnosability under  Sensor Constraints}\label{sec:Diagnosability under  Sensor Constraints}
In this section, we investigate diagnosability under sensor constraints.
Specifically, we formally propose the notion of $\varphi$-diagnosability as the necessary and sufficient condition of the existence of a diagnoser working correctly for DES subject to sensors  constrained by specification $\varphi$.

First, we modify the existing definition of diagnosability in Definition~\ref{def:classical diag} to a new notion of diagnosability, called  \emph{$\varphi$-diagnosability}, by taking the issues of the LTL sensor constraint $\varphi$ into account. 

\begin{mydef}[$\varphi$-Diagnosability]\label{def:varphi-diagnosability}\upshape 
	System $G$ is said to be $\varphi$-diagnosable w.r.t.\  output function $\mathcal{O}$, fault events $\Sigma_F$ and  sensor constraint $\varphi$ if and only if
	\begin{equation} \label{eq:varphi-diagnosis}
		(\forall s \in \mathcal{L}^{\varphi}_{e,F}(G))(\exists t \in \overline{\{s\}})[ \varphi \textsf{-diag} ]  
	\end{equation}
	where the $\varphi$-diagnostic condition $\varphi \textsf{-diag}$ is 
	\[
	(\forall \omega \in \overline{\mathcal{L}^{\varphi}_{e}(G)} )  
	[\Theta_{\Delta}(w)=\Theta_{\Delta}(t) \Rightarrow    \Sigma_{e,F}\in w].  
	\]  
\end{mydef} 

The above definition says that, for any infinite faulty extended string $s$ satisfying sensor constraint $\varphi$, there is a finite prefix $t$ such that for any finite extended string $\omega$ that are possible as a prefix of some infinite string satisfying the sensor constraint, if the outputs of $\omega$ and $t$ are equivalent, then the extended string $\omega$ must also be faulty. 
Intuitively, $\varphi$-diagnosability modifies the standard diagnosability by restricting our attention only to those infinite extended strings  satisfying the sensor constraint $\varphi$.  
Clearly, by setting $\varphi=\texttt{true}$, our $\varphi$-diagnosability becomes to the standard diagnosability in Definition~\ref{def:classical diag}.

\begin{remark}\upshape
	In Definition \ref{def:classical diag}, it is known that 
	``$\forall s \!\in\! \Psi_e(G)$" and  ``$\exists n\!\in\! \mathbb{N}$"
	can be swapped for finite-state automata, which means that if the system is diagnosable, then there exists a \emph{uniform} detection bound after the occurrence of any fault event. 
	Here since the assumption of LTL sensor constraint is imposed on infinite strings, 
	we can only guarantee that for any infinite faulty string, there exists a finite detection bound. 
	However, this does not imply that there exists a uniform detection bound for all faulty strings. 
	For example, suppose that there is an indicator event that will occur infinite number of times after fault events. If we assume that the sensor constraint $\varphi$ for the indicator event is that it will \emph{not always fail}, i.e., its occurrence will \emph{eventually} be observed correctly, then the system is $\varphi$-diagnosable. However, this condition does not ensure when it will be observed since the satisfaction instant of ``eventually" can be arbitrarily late.  	 
\end{remark}

We use  the following example to illustrate   the notion of $\varphi$-diagnosability. 

\begin{myexm}\label{exm:varphi-Diagnosability}\upshape
	Again, let us consider system $G_1$ shown in Figure~\ref{fig:System G_1} with the same setting in Example~\ref{exam:permanent-specification}. 
	Note that there exists an infinite faulty extended string 
	\[
	s_F=(1,f, \varepsilon)((2,a, a)(3,b,\varepsilon)(4,c,c))^\omega,
	\]
	whose trace  is
	$\textsf{trace}(s_F)=\emptyset ( \emptyset \{m_1\}   \emptyset  )^\omega \models \varphi_{\textsf{per}}$, 
	i,e., $s_F \in \mathcal{L}^{\varphi}_{e,F}(G)$. However,  there exists an infinite normal extended string
	\[s_N=(1,u, \varepsilon)(2,a,a)(3,b,\varepsilon)((4,c,c)(5,a,a))^\omega,\]
	whose trace is
	$\textsf{trace}(s_N)=\emptyset \emptyset \{m_1\} (\emptyset \emptyset)^\omega \models \varphi_{\textsf{per}}$, i.e.,   $s_N \in \mathcal{L}^{\varphi}_{e}(G)$, such that the output of $s_F$ and $s_N$ are the same, i.e., $\Theta_{\Delta}(s_F)=\Theta_{\Delta}(s_N)=(a c)^\omega$. 
	Therefore, for any finite prefix $t \in \overline{\{s_F\}}$, there exists a normal extended string $\omega \in \overline{\{s_N\}}$ such that $\Theta_{\Delta}(t)=\Theta_{\Delta}(\omega)$. 
	Namely, there exists an infinite faulty extended string compatible to the sensor constraint $\varphi_{\textsf{per}}$ such that for each of its prefix, there exists a normal extended string having the same output with the prefix. 
	By Definition \ref{def:varphi-diagnosability}, system $G_1$ is not $\varphi_{\textsf{per}}$-diagnosable.
\end{myexm}

Next, we show that the proposed notion of  $\varphi$-diagnosability indded provides the necessary and sufficient condition for the existence of a diagnoser that works ``correctly" under sensor constraint $\varphi$. 
Formally, a diagnoser is a function
\[
D: \mathcal{O}(\mathcal{L}(G)) \to \{0,1\}
\]
that decides whether a fault has happened (by issuing ``$1$") or not (by issuing ``$0$") based on the output string. 
We say that a diagnoser works correctly under the sensor constraint $\varphi$ if it satisfies the following conditions:
\begin{enumerate}[C1)]
	\item 
	The diagnoser will evetually issue a fault alarm for any occurrence of fault events, i.e.,
	\[
	(\forall s  \in \mathcal{L}^{\varphi}_{e,F}(G))(\exists t \in \overline{\{s\}}))[D(\Theta_{\Delta}(t))=1].\label{C1}
	\]
	\item 
	The diagnoser will not issue a false alarm if the execution is still normal, i.e., 
	\[
	(\forall s \in \overline{\mathcal{L}^{\varphi}_{e}(G)}: \Sigma_{e,F} \notin s) [D(\Theta_{\Delta}(s)) = 0].
	\label{C2}
	\]
\end{enumerate}

The following theorem says that there exists a diagnoser which works ``correctly" under sensor constraint $\varphi$ if and only if the system is  $\varphi$-diagnosable.
\begin{mythm}\label{thm:varphi-existence}
	There exists a diagnoser  satisfying conditions C\ref{C1} and C\ref{C2} 
	if and only if $G$ is $\varphi$-diagnosable w.r.t.\  fault events $\Sigma_F$, output function $\mathcal{O}$ and sensor  constraint $\varphi$.
\end{mythm}
\begin{proof}
	All proofs hereafter are provided in the Appendix.  
\end{proof} 

\section{Verification of $\varphi$-diagnosability}\label{sec:Verification of varphi-diagnosability}
In this section, we focus on the verification of $\varphi$-diagnosability. 
Specifically, we provide a verifiable necessary and sufficient condition for $\varphi$-diagnosability based on the verification system.  

\subsection{Augmented System} 
To verify $\varphi$-diagnosability, our first step is to augment both the state-space and the event-space of $G$ such that 
\begin{itemize}
	\item 
	the information of whether or not a fault event has occurred is encoded in the \emph{augmented state-space}; and 
	\item 
	the information of which state the event is enabled from and which specific output is observed are encoded in the \emph{augmented event-space}. 
\end{itemize} 
\begin{mydef}[Augmented Systems]\upshape
Given  system $G = (Q,q_0,\Sigma,\delta)$, fault events $\Sigma_F$ and output function $\mathcal{O}$, we define the \emph{augmented system} as a new-tuple
\begin{equation} \label{eq:Augment System}
	\tilde{G}=(\tilde{Q},\tilde{q}_0,\Sigma_e,\tilde{\delta}),
\end{equation}
where 
\begin{itemize}
	\item 
	$\tilde{Q} \subseteq Q \times \{F,N\}$ is the set of augmented states;
	\item 
	$\tilde{q}_0=(q_0,N)$ is the initial augmented state;
	\item 
	$\Sigma_e$ is the set of augmented events, which are just extended events;
	\item 
	$\tilde{\delta}:\tilde{Q} \times \Sigma_e \rightarrow \tilde{Q}$ is the transition function defined by: 
	for any $\tilde{q}=(q,l)\in \tilde{Q}$ and $\tilde{\sigma}=(q,\sigma,o) \in \Sigma_e$, we  have 
	$\tilde{\delta}(\tilde{q},\tilde{\sigma})!$ whenever  	$f(q,\sigma)!$ and $o\in \mathcal{O}(q,\sigma)$. 
	Furthermore, when $\tilde{\delta}(\tilde{q},\tilde{\sigma})!$, we have  
	\[
	\tilde{\delta}(\tilde{q},\tilde{\sigma})=
	\left\{ 
	\begin{array}{l l}
		(\delta(q,\sigma),N) &\text{ if } l=N\wedge \tilde{\sigma} \!\notin\! \Sigma_{e,F} \\ 
		(\delta(q,\sigma),F) &\text{ otherwise }
	\end{array} 
	\right..  
	\] 
\end{itemize}
\end{mydef} 

The above constructed augmented system $\tilde{G}$ has the following properties:
\begin{itemize}
	\item 
	First, the augmented system $\tilde{G}$ generates extended strings. 
	Essentially, it still tracks the original dynamic of the system by putting both the output realization and the current state information together with the internal event. 
	Therefore, we have $\mathcal{L}(\tilde{G})=\mathcal{L}_e(G)$.
	\item 
	Second, each augmented state $(q,l) \in \tilde{Q}\subseteq Q \times \{F,N\}$ has two components. 
	The first component $q$ is the actual state in the original system $G$ and the second component $l \in \{N,F\}$ is a label denoting whether fault events have occurred. 
	By the construction, 
	the label will change from $N$ to $F$ only when an extended fault event occurs and 
	once the label becomes $F$, it will be $F$ forever. 
	We denote by $\tilde{Q}_N=\{(q,l)\in \tilde{Q}: l = N\}$ the set of normal augmented states and the set of faulty states  $\tilde{Q}_F$ is defined analogously.  
\end{itemize}

\begin{myexm}\upshape
	Still, we consider system $G_1$ in Figure~\ref{fig:System G_1}, which has been   discussed in Example \ref{exam:permanent-specification}. 
	Its augmented system $\tilde{G}_1$ of $G_1$ is  depicted in Figure~\ref{fig:Augment system G_1}, 
	where all states reachable via extended fault event $(1,f, \varepsilon)$  are augmented with label $F$ and we denote every extended event $(q,\sigma,o)\in \Sigma_e$ by $\sigma^o_q$. 
	Furthermore, the transitions of  $\tilde{G}_1$ are defined according to the actual transitions and the underlying observations of the original system $G_1$. For example, because $\delta(3,b)=4$ and $\mathcal{O}(3,b)=\{b,\varepsilon\}$, we have two new transitions in $\tilde{G}_1$: $\tilde{\delta}(3F, (3,b,b))=4F$ and $\tilde{\delta}(3F, (3,b,\varepsilon))=4F$; each for them represents a different observation realization. 
\end{myexm}

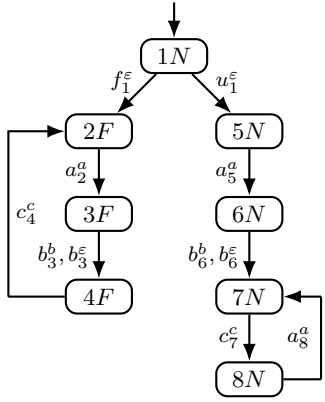
\begin{figure} 
	\centering
	\begin{tikzpicture}[->,>={Latex}, thick, initial text={}, initial where=above, thick, node distance=1.1cm, base node/.style={rectangle, align = center, draw, minimum height=3mm, rounded corners =1.5mm, font=\small}]		
		\node[initial, state, base node, ] (1) at (1.5,5) {$1N$}; 
		\node[state, base node, ](2) at (0.5,4) {$2F$};
		\node[state, base node, ](3)  [below of=2] {$3F$};
		\node[state, base node, ](4)  [below of=3] {$4F$};
		\node[state, base node, ](5) at (2.5,4) {$5N$};
		\node[state, base node, ](6)  [below of=5] {$6N$};
		\node[state, base node, ](7)  [below of=6] {$7N$};
		\node[state, base node, ](8)  [below of=7] {$8N$};
		
		\path[]
		(1) edge node [left, xshift=0.1cm, yshift=0.15cm] {\fontsize{8}{1} ${f}_1^\varepsilon$} (2)
		(1) edge node [right, xshift=-0.2cm, yshift=0.15cm] {\fontsize{8}{1} $u_{1}^\varepsilon$} (5)
		(2) edge node [left] {\fontsize{8}{1} $a_{2}^a$} (3)
		(3) edge node [left] {\fontsize{8}{1} $b_{3}^b, b_{3}^\varepsilon$} (4);
		
		\draw (4) -- ($(4.west)+(-0.75,0)$)
		($(4.west)+(-0.75,0)$) -- node [midway,right, xshift=-0.15cm] {\fontsize{8}{1} $c_{4}^c$} 
		($(2.west)+(-0.75,0)$)
		-- (2.west);
		
		\path[]
		(5) edge node [midway,left] {\fontsize{8}{1} $a_{5}^a$} (6)
		(6) edge node [midway,left] {\fontsize{8}{1} $b_{6}^b,b_{6}^\varepsilon$} (7)
		(7) edge node [left ] {\fontsize{8}{1} $ c_{7}^c$} (8);
		
		\draw
		(8) -- ($(8.east)+(0.5,0)$)
		($(8.east)+(0.5,0)$) -- node [midway,left] {\fontsize{8}{1} $a_{8}^a $} ($(7.east)+(0.5,0)$)
		($(7.east)+(0.5,0)$) -- (7.east);
		
	\end{tikzpicture}
		\caption{\label{fig:Augment system G_1}Augment system $\tilde{G}_1$ of system $G_1$. For simplicity, we also use $\sigma_q^o$ to denote the extended event $(q, \sigma, o) \in \Sigma_e$.}
	\end{figure}

\subsection{Observation Constrained  System} 
Recall that for any LTL formula $\varphi$, there exists an NBA $\mathcal{B}=(X, X_0,2^\mathcal{AP},\xi, X_m)$ such that $\mathcal{L}^\omega_m(\mathcal{B})=\textsf{word}(\varphi)$.
In order to capture all possible extended strings that can be observed under sensor constraint $\varphi$, we construct the  observation constrained system. 

\begin{mydef}[Observation Constrained System]\upshape
Given  augmented system $\tilde{G}=(\tilde{Q},\tilde{q}_0,\Sigma_e,\tilde{\delta})$
and NBA $\mathcal{B}=(X, X_0,2^\mathcal{AP},\xi, X_m)$ associated with LTL formula $\varphi$, the  \emph{observation constrained system} is defined as a new-tuple
\begin{equation} \label{eq:CPS}
	T=(Q_T,\Sigma_e,\delta_T,Q_{0,T},Q_{m,T}),
\end{equation}
where 
\begin{itemize}
	\item 
	$Q_T \subseteq \tilde{Q} \times X$ is the set of states;
	\item 
	$\Sigma_e$ is still the set of extended events;
	\item 
	$\delta_T:Q_T \times \Sigma_e \rightarrow 2^{Q_T}$ is the non-deterministic transition function defined by: 
	for any $q_T=(\tilde{q},x) \in Q_T$ and $\sigma_e \in \Sigma_e$, we have
	\[
	\delta_T(q_T,\sigma_e)=
	\left\{
	(\tilde{q}',x') : \!\!\!\!\!
	\begin{array}{cc}
		& \tilde{q}'= \tilde{\delta}(\tilde{q},\sigma_e) \text{ and }\\
		&  x' \in \xi(x,\textsf{label}(\sigma_e))
	\end{array}
	\right\}
	\] 
	\item 
	$Q_{0,T}=\{\tilde{q}_0\} \times X_0$ is the set of initial states;
	\item
	$Q_{m,T}=\{ (\tilde{q},x) \in Q_T: x \in X_m\}$ is the set of accepting states.
\end{itemize}
\end{mydef}

Intuitively, the observation constrained system $T$ is constructed by synchronizing the augmented system $\tilde{G}$ with the NBA  $\mathcal{B}$ associated with $\varphi$ according to the atomic propositions. 
Specifically, for any states $q_T=(\tilde{q},x),q_T'=(\tilde{q}',x') \in \tilde{Q} \times X$ and extended event $\sigma_e \in \Sigma_e$, we have $q'_T \in \delta_T(q_T, \sigma_e)$ whenever
(i) in the  first  (plant model) component, the event itself satisfies the dynamic of the system, i.e., $\tilde{q}'=\tilde{\delta}(\tilde{q},\sigma_e)$; and 
(ii) in the second (LTL formula) component, the atomic propositions of the event satisfies the transition rules of the B\"{u}chi automaton, i.e., 
$x' \in \xi(x,\textsf{label}(\sigma_e))$.  
Therefore, for any  string $s=\sigma_1\sigma_2\cdots \sigma_n\in  \mathcal{L}(T)$, 
we have $s\in \mathcal{L}_e(G)=\mathcal{L}(\tilde{G})$ 
and $\textsf{trace}(s) \in \mathcal{L}(\mathcal{B})$. 
Furthermore, by construction, a state $(\tilde{q},x) \in Q_T$ is accepting iff its second  component $x\in X_m$ is  accepting state in NBA $\mathcal{B}$. 
Therefore,  $T$ essentially recognizes all infinite strings in $G$ satisfying the sensor constraint $\varphi$, i.e., 
\begin{equation}
	\mathcal{L}^\omega_m(T)= \mathcal{L}^{\varphi}_{e}(G).
\end{equation} 

In Definition \ref{def:varphi-diagnosability}, we need to compare the observations of  prefixes of those strings in $\mathcal{L}^{\varphi}_{e,F}(G)$ 
with finite strings in $\overline{\mathcal{L}^{\varphi}_{e}(G)}$.  
Using observation constrained system $T$, we know that 
finite string $s\in \overline{\mathcal{L}^{\varphi}_{e}(G)}$ iff its trace can reach a state in $T$ from which accepting states can be visited infinitely often. 
To this end, we say a state $q\in Q_T$ in $T$ is \emph{feasible} if 
$\mathcal{L}_m^\omega( T_q )\neq \emptyset$, 
where $T_q=(Q_T,\Sigma_e,\delta_T,\{q\},Q_{m,T}) $ is the NBA by setting the initial state of $T$ as state $q$. 
We denote by $Q_{feas,T}\subseteq Q_T$ the set of feasible states. 
Then we  have the following equivalence
\begin{equation}\label{eq:feasible}
 s \in \overline{\mathcal{L}^{\varphi}_{e}(G)} 
 \Leftrightarrow
 \exists q_0 \in Q_{0,T},
 \delta_T( q_0,s) \cap Q_{feas,T}   \neq \emptyset.   \neq \emptyset.
\end{equation} 

Also, we denote by $Q_{N,T}=\{ (q, N, x) \in Q_T \}$ and $Q_{F,T}=\{ (q, F, x) \in Q_T \}$ as the set of normal and faulty states in $T$, respectively. 
Then for any $s \in \overline{\mathcal{L}^{\varphi}_{e}(G)}$, we also have the following equivalence
\begin{equation}
	\Sigma_{e,F}\in s  
	\Leftrightarrow \exists q_0 \in Q_{0,T},
	\delta_T( q_0,s) \cap Q_{F,T} \neq \emptyset.
\end{equation} 
We illustrate the above concepts by the following example.

\begin{myexm}\upshape
	Let us still consider system $G_1$ in Figure~\ref{fig:System G_1}, whose augmented system $\tilde{G}_1$ has been given in Figure~\ref{fig:Augment system G_1}. 
	We still consider sensor constraint $\varphi_{\textsf{per}}$ in Equation~\eqref{eq:permanent equation}, 
	which can be translated to NBA $\mathcal{B}_{\textsf{per}}$ as shown in Figure~\ref{fig:Permanent specification}. 
	Based on   $\tilde{G}_1$ and  $\mathcal{B}_{\textsf{per}}$, we construct the observation constrained system 
	$T_1$ as shown in Figure~\ref{fig:observation constrained system-permanent}, where we omit double circles as all states are accepting.  
	Since all states are accepting and the system is live, 
	we know that all states in $T_1$ are feasible. 
	For example, finite string
	$s=(1,f,\varepsilon)(2,a,a)(3,b,\varepsilon)(4,c,c)(2,a,a)(3,b,b)\in \mathcal{L}(\tilde{G})$, 
	we have $s\notin \mathcal{L}(T)$ 
	since proposition $m_1$ has already been satisfied by $(3,b,\varepsilon)$  and 
	proposition $m_0$ cannot hold thereafter, i.e., extended event $(3,b,b)$ cannot be synchronized with  $\mathcal{B}_{\textsf{per}}$ when constructing $T_1$.  
	Hence, we know that $s\notin\overline{\mathcal{L}^\varphi_e(G)}$.
\end{myexm}

\begin{figure} 
	\centering
	\centering
	\begin{tikzpicture}[->,>=stealth,shorten >=1pt,auto,node distance=2.3cm,
		thick,base node/.style={circle,draw,minimum size=0.5mm}, font=\small]
		\node[initial,initial text={}, state, accepting,base node] (A) {$A$};
		\node[state, accepting,base node](B)[right of=A ]{$B$}[above];
		
		\path[]
		(A) edge [loop above] node {\fontsize{8}{1}$\neg m_1$} (A)
		(A) edge node {\fontsize{8}{1}$\neg m_0 \wedge m_1$} (B)
		(B) edge [loop above] node {\fontsize{8}{1}$\neg m_0$} (B);
	\end{tikzpicture}
		\caption{\label{fig:Permanent specification}NBA $\mathcal{B}_{\textsf{per}}$ for specification $\varphi_{\textsf{per}}=\Box( m_1 \to \Box \neg m_0)$ with $\mathcal{AP}=\{m_0,m_1\}$, where accepting states are highlighted by double circles. 
			Here, we follow  the standard abbreviation for drawing NBA over alphabet $2^{\mathcal{AP}}$. 
			For example,  transition $\neg m_1$ is the abbreviation of $\{\emptyset, m_0\}$ and 
			transition $\neg m_0 \wedge m_1$ is the abbreviation of $\{ m_1\}$.}
\end{figure}
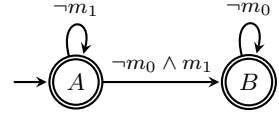
	
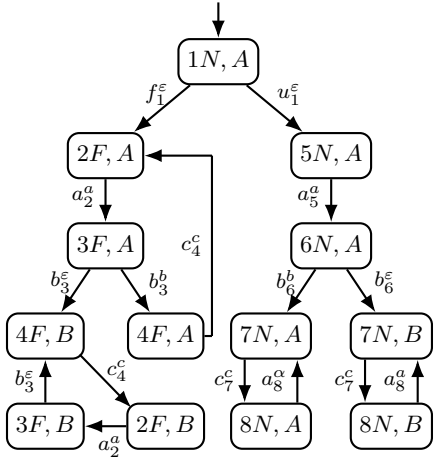
\begin{figure} 
	\centering
	\begin{tikzpicture}[->,>={Latex}, thick, initial text={}, node distance=1.2cm, initial where=above, thick, base node/.style={rectangle, align = center, draw, minimum height=6mm, rounded corners =1.5mm, font=\small }]		
	\node[initial, state, base node] (1) at (2.5,7) {$1N,A$}; 
	\node[state, base node](2) at (1,5.75) {$2F,A$};
	\node[state, base node](3) [below of=2] {$3F,A$};
	\node[state, base node](4-B) [below of=3, xshift=-0.8cm] {$4F,B$};
	\node[state, base node](4-A) [below of=3, xshift=0.8cm] {$4F,A$};
	\node[state, base node](3-B) [below of=4-B] {$3F,B$};
	\node[state, base node](2-B) [below of=4-A] {$2F,B$};
	\node[state, base node](5) at (4,5.75) {$5N,A$};
	\node[state, base node](6) [below of=5] {$6N,A$};
	\node[state, base node](7) [below of=6, xshift=-0.8cm] {$7N,A$};
	\node[state, base node](8) [below of=7] {$8N,A$};
	\node[state, base node](7-B) [below of=6, xshift=0.8cm] {$7N,B$};
	\node[state, base node](8-B) [below of=7-B] {$8N,B$};
	\path[]
	(1) edge node [left, pos=0.2] {\fontsize{8}{1} $f_1^\varepsilon$} (2)
	(1) edge node [right, pos=0.2] {\fontsize{8}{1} $u_{1}^\varepsilon$} (5)
	(2) edge node [left, pos=0.4] {\fontsize{8}{1} $a_{2}^a$} (3)
	(3) edge node [right, pos=0.3] {\fontsize{8}{1} $b_{3}^b$} (4-A)
	(3) edge node [left, pos=0.3] {\fontsize{8}{1} $b_{3}^\varepsilon$} (4-B)
	(3-B) edge node [left, pos=0.5] {\fontsize{8}{1} $b_{3}^\varepsilon$} (4-B);
	\draw ($(4-B.south east)+(-0.068,0.068)$)--node [right, pos=0.3, xshift=-0.1cm, yshift=-0.1] {\fontsize{8}{1} $c_{4}^c$} ($(2-B.north west)+(0.068,-0.068)$);
		
	\draw (2-B) -- node [below, pos=0.5] {\fontsize{8}{1} $a_{2}^a$} (3-B);
			
	\draw (4-A) -- ($(4-A.east)+(0.1,0)$)
	($(4-A.east)+(0.1,0)$)-- node [left, pos=0.5] {\fontsize{8}{1} $c_{4}^c$} ($(4-A.east)+(0.1,0) + (5) -(6)+(5) -(6)$)
	($(4-A.east)+(0.1,0) + (5) -(6)+(5) -(6)$)--(2);
			
	\path[]
	(5) edge node [left, pos=0.4] {\fontsize{8}{1} $a_{5}^a$} (6);
	\path[]
	(6) edge node [left, pos=0.3] {\fontsize{8}{1} $b_{6}^b$} (7)
	($(7.south)-(0.35,0)$) edge node [left, pos=0.5] {\fontsize{8}{1}  $c_{7}^c$} ($(8.north)-(0.35,0)$)
	($(8.north)+(0.35,0)$) edge node [left, pos=0.5] {\fontsize{8}{1} $a_{8}^\alpha$} ($(7.south)+(0.35,0)$);
	\path[]
	(6) edge node [right, pos=0.3] {\fontsize{8}{1} $b_{6}^\varepsilon$} (7-B)
	($(7-B.south)-(0.35,0)$) edge node [left, pos=0.5] {\fontsize{8}{1} $ c_{7}^c$} ($(8-B.north)-(0.35,0)$)
	($(8-B.north)+(0.35,0)$) edge node [left, pos=0.5] {\fontsize{8}{1} $a_{8}^a$} ($(7-B.south)+(0.35,0)$);
	\end{tikzpicture}
	\caption{\label{fig:observation constrained system-permanent} The observation constrained system $T_1$ for system $\tilde{G_1}$ with sensor  constraint $\varphi_{\textsf{per}}$, where all states are accepting.}
\end{figure}

\subsection{Verification Structure}	
According to Definition \ref{def:varphi-diagnosability}, a system is \emph{not} $\varphi$-diagnosable iff there exists an infinite extended faulty string $s \in \mathcal{L}^{\varphi}_{e,F}(G)$ such that for any prefix $t \in \overline{\{s\}}$, there is an extended normal string $\omega \in \overline{\mathcal{L}^\varphi_e(G)}$ having the same output with $t$. 
Motivated by this observation, we construct the  verification system  to capture all pairs of extended faulty strings and extended normal strings  that have the same  outputs and both satisfy the sensor constraint.  

\begin{mydef}[Verification System]\upshape 
Given system $G$ and sensor constraint $\varphi$,  
its verification system   is a new tuple 
\begin{equation}\label{eq:verification system}
	V=(Q_V,\Sigma_V,\delta_V,Q_{0,V},Q_{m,V}),
\end{equation}
where
\begin{itemize}
	\item 
	$Q_V \subseteq Q_T \times Q_T$ is the finite set of states;
	\item 
	$\Sigma_V=\Sigma^{o}_V \cup \Sigma^{uo}_V$ is the finite set of events, where 
	\begin{itemize}
		\item 
		$\Sigma^{o}_V=\{(\sigma_1,\sigma_2) \!\in\! \Sigma_e \!\times\! \Sigma_e: \Theta_{\Delta}(\sigma_1)\!=\!\Theta_{\Delta}(\sigma_2) \!\neq\! \varepsilon\}$;
		\item $\Sigma^{uo}_V=\{(\sigma_1,\varepsilon) \!\in\! \Sigma_e \!\times\! \{ \varepsilon\}: \Theta_{\Delta}(\sigma_1)\!=\!\varepsilon \} \cup \{(\varepsilon,\sigma_2) \!\in \!\{ \varepsilon\} \!\times\! \Sigma_e: \Theta_{\Delta}(\sigma_2)\!=\!\varepsilon \}$;
	\end{itemize}
	\item 
	$\delta_V:Q_V \times \Sigma_V \rightarrow 2^{Q_V}$ is the non-deterministic transition function defined by:  	for any $q_V=(q_1,q_2) \in Q_V$ and $\sigma_V=(\sigma_1, \sigma_2) \in \Sigma_V$, we have	
	\[
	\delta_V(q_V,\sigma_V)=   \delta_T(q_1,\sigma_1)  \times  \delta_T(q_2,\sigma_2); 
	\] 
	\item 
	$Q_{0,V}= Q_{0,T} \times Q_{0,T}$ is the set of initial states;
	\item 
	$Q_{m,V} \subseteq Q_V$ is the set of accepting states defined by 
	\[
	Q_{m,V}=\left\{( q_1 , q_2 ) \in Q_V:  \!\!\!
	\begin{array}{c c}
		q_1 \in    Q_{m,T} \cap Q_{F,T}  \\
	\text{and }	q_2 \in   Q_{feas,T}\cap Q_{N,T}
	\end{array}\!\!
	\right\}.
	\] 
\end{itemize}
\end{mydef} 

Intuitively, 
each state $q_V=(q_1, q_2)$ in the verification system $V$ is a pair of states in system $T$.  
The event set $\Sigma_V$  is divided into two categories: $\Sigma_V=\Sigma^{o}_V \cup \Sigma^{uo}_V$: 
event $(\sigma_1,\sigma_2)$ is in $\Sigma^{o}_V$ if both $\sigma_1$ and $\sigma_2$ have the same non-empty output 
and
event $(\sigma_1,\sigma_2)$  is  in $\Sigma^{uo}_V$  if  the output of $\sigma_1 (\sigma_2)$ is empty and $\sigma_2 (\sigma_1)$ is empty.  
Essentially, $V$ is obtained by synchronizing $T$ with its copy according to their outputs.  
Then for event $\sigma_V=(\sigma_1,\sigma_2)$, we 
denote by  $\theta_1(\sigma_V)=\sigma_1$ and $\theta_2(\sigma_V)=\sigma_2$ its first and second components, respectively; 
the notation is also   extended to a string $s=\sigma^1_V \sigma^2_V \cdots \in \Sigma^*_V \cup \Sigma^\omega_V$ 
by 
$\theta_1(s)=\theta_1(\sigma^1_V) \theta_1(\sigma^2_V)\cdots$ and
$\theta_2(s)=\theta_2(\sigma^1_V) \theta_2(\sigma^2_V)\cdots$.
By the construction, $T$ has the following properties: 
\begin{itemize}
	\item 
	For any $s \in \mathcal{L}(V)$, 
	we have   $s_1=\theta_1(s), s_2=\theta_2(s)\in \mathcal{L}(T)$ and  $\Theta_{\Delta}(s_1)=\Theta_{\Delta}(s_2)$;
	\item 
	For any pair of extended strings $s_1,s_2 \in \mathcal{L}(T)$ such that $\Theta_{\Delta}(s_1)=\Theta_{\Delta}(s_2)$, 
	there exits a string $s \in \mathcal{L}(V)$ such that $\theta_1(s)=s_1$ and $\theta_2(s)=s_2$.
\end{itemize}

The accepting conditions in  $Q_{m,V}$ are explained as follows. 
A state $q_V=(q_1,q_2)$ is accepting if
(i) its first component $q_1$ is a  faulty accepting state in $T$; 
and 
(ii) its second component $q_2$ is normal and feasible. 
The first condition says that by repeatedly visiting state $q_V\in Q_{m,V}$, 
the first component of the string will be an infinite faulty string satisfying $\varphi$. 
The second condition says that for any finite string reaching $q_V\in Q_{m,V}$, the second component of the string is a normal string in 
$\overline{ \mathcal{L}^\varphi_e(G) }$.  

\subsection{Checking $\varphi$-Diagnosability}
Now, we present how to verify $\varphi$-diagnosability using the verification system $V$.  To this end, we first introduce some concepts. 
A \emph{run} in $V$ is a finite sequence 
\begin{equation}\label{eq:run}
	\pi=q_V^1 \xrightarrow{\sigma_V^1} q_V^2 \xrightarrow{\sigma_V^2} \cdots \xrightarrow{\sigma_V^{n-1}} q_V^n, 
\end{equation}
where $q_V^i \in Q_V, \sigma_V^i \in \Sigma_V$ and $q_V^{i+1} \in \delta_V(q_V^i,\sigma_V^i)$.  
A run  $\pi$  of Equation~\eqref{eq:run} is called a \emph{cycle} if $q_V^1=q_V^n$;  
a cycle $\pi$ is said to be \emph{reachable} if  there exists a finite string $s \in \mathcal{L}(V)$ and $q_{0,V} \in Q_{0,V}$ such that $q_V^1 \in f_V(q_{0,V},s)$.  

We are now ready to present the main  theorem for verifying $\varphi$-diagnosability.

\begin{mythm}\label{thm:verify varphi-diag}
System $G$ is not $\varphi$-diagnosable w.r.t.\ fault events $\Sigma_F$, output function $\mathcal{O}$ and sensor constraint $\varphi$, if and only if, in the verification system $V$, there exists a reachable cycle
	\[
	\pi = q_V^1 \xrightarrow{\sigma_V^1} q_V^2 \xrightarrow{\sigma_V^2} \cdots \xrightarrow{\sigma_V^{n-1}} q_V^n
	\]
	such that  
	\begin{enumerate}
		\item 
		$\theta_1(\sigma_V^i) \neq \varepsilon$  for some $i=1,\dots,n$; and 
		\item 
		$q_V^{j} \in Q_{m,V}$ for some $j=1,\dots,n$. 
	\end{enumerate} 
\end{mythm}

The intuition of the above theorem is as follows. 
Since $\pi$ is a reachable cycle, we know that 
we can find an infinite string  $s=t(\sigma_V^1\cdots \sigma_V^N)^\omega$ in $V$, 
where $t$ is some string entering the cycle.  
Then the two conditions ensure that 
(i) its first component $\theta_1(s)$ is indeed an
infinite extended string in $\mathcal{L}(\tilde{T})$ that is both faulty and satisfying $\varphi$; and 
(ii) its second component $\theta_2(s)$ is an extended string 
such that any of its finite prefix is normal and can be extended to an infinite string satisfying $\varphi$, i.e., $\theta_2(s) \in \overline{\mathcal{L}^{\varphi}_{e}(G)}$.  
Moreover, the construction of $V$ ensures that $\theta_1(s)$ and $\theta_2(s)$ have the same output. 
Therefore, the existence of such a cycle falsifies $\varphi$-diagnosability. 
The detailed proof is provided in the Appendix. Here we illsurate this theorem by the following example. 

\begin{myexm} \label{exm:checking permanent}\upshape
Let us still consider the running example   $G_1$ in Figure~\ref{fig:System G_1} with the same setting in Example \ref{exam:permanent-specification}. 
As we have discussed in Example \ref{exm:varphi-Diagnosability}, this system is not $\varphi_{\textsf{per}}$-diagnosable.  Here, we analyze this more formally by Theorem \ref{thm:verify varphi-diag}. 
Based on its observation constrained system $T_1$, we construct the verification system $V_1$, which is partially  shown in Figure~\ref{fig:Verfication-Per}, where accepting states are marked by double circles, e.g., state $\{(4F,B),(7N,B)\}$ is  in $Q_{m,V}$  as $(4F,B) \in Q_{m,T} \cap Q_{F,T}$ and $(7N,B) \in Q_{feas,T}\cap Q_{N,T}$.
	Here, we only focus on the reachable cycle satisfying conditions in Theorem \ref{thm:verify varphi-diag} and omit other parts without loss of generality for the purpose of verification. 
	Specifically, we consider the cycle as highlighted with red lines in the Figure~\ref{thm:verify varphi-diag}, i.e.,
	\begin{align*}
		\{(4F,B),(7N,B)\} \xrightarrow{(c_4^c,c_7^c)} \{(2F,B),(8N,B)\} \xrightarrow{(a_2^a,a_8^a)}\\
		\{(3F,B),(7N,B)\} \xrightarrow{(b_3^\varepsilon,\varepsilon)} \{(4F,B),(7N,B)\}.
	\end{align*} 
	Note that all states in the cycle are accepting states and there exist  events $(c_4^c,c_7^c)$ and $(a_2^a,a_8^a)$ such that $\theta_1((c_4^c,c_7^c))\neq \varepsilon$ and $\theta_1((a_2^a,a_8^a))\neq \varepsilon$. Thus, system $G_1$ is not $\varphi_{\textsf{per}}$-diagnosable according to Theorem \ref{thm:verify varphi-diag}. 
\end{myexm}

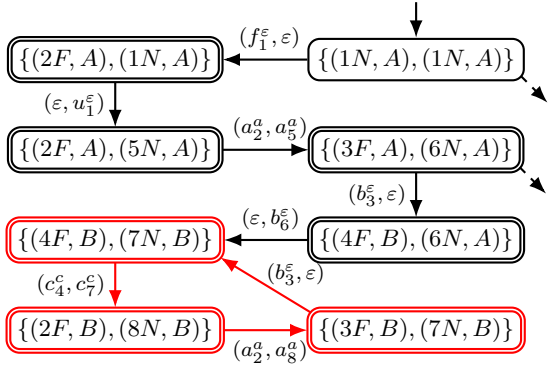
\begin{figure} 
	\centering
	\begin{tikzpicture}[->,>={Latex}, thick, initial text={},node distance=1.2cm, initial where=above, thick, base node/.style={rectangle, align = center, draw, minimum height=0mm, rounded corners =1.5mm, font=\small}, ]		

		\node[initial, state, base node,] (1-2) at (4,5) {$\{(1N,A),(1N,A)\}$}; 
		\node[state,accepting, base node](2-2) [below of=1-2] {$\{(3F,A),(6N,A)\}$};
		\node[state, accepting,base node](3-2) [below of=2-2] {$\{(4F,B),(6N,A)\}$};
		\node[state,accepting, base node,red](4-2) [below of=3-2] {\color{black}$\{(3F,B),(7N,B)\}$};
		
		\node[state,accepting, base node](1-1) at (0,5) {$\{(2F,A),(1N,A)\}$};
		\node[state,accepting, base node](2-1) [below of=1-1] {$\{(2F,A),(5N,A)\}$};
		
		\node[state,accepting, base node,red](3-1) [below of=2-1] {\color{black}$\{(4F,B),(7N,B)\}$};
		\node[state,accepting, base node,red](4-1) [below of=3-1] {\color{black}$\{(2F,B),(8N,B)\}$};
		
		\draw (1-2) -- (1-1) node [pos=0.5, above] {\fontsize{8}{1} $({f}_1^\varepsilon, \varepsilon)$};
		\draw[dashed] ($(1-2.south east)+(-0.05,0.068)$) -- +(-45:0.5cm);
		
		\draw (1-1) -- (2-1) node [pos=0.5, left] {\fontsize{8}{1} $(\varepsilon,u_1^\varepsilon)$};
		\draw (2-1) -- (2-2) node [pos=0.5, above] {\fontsize{8}{1} $(a_2^a,a_5^a)$};
		\draw[dashed] ($(2-2.south east)+(-0.05,0.068)$) -- +(-45:0.5cm);
		
		\draw (2-2) -- (3-2) node [pos=0.5, left] {\fontsize{8}{1} $(b_3^\varepsilon,\varepsilon)$};
		\draw (3-2) -- (3-1) node [pos=0.5, above] {\fontsize{8}{1} $(\varepsilon,b_6^\varepsilon)$};
		\draw[red] (3-1) -- (4-1) node [pos=0.5, left,black] {\fontsize{8}{1} $(c_4^c,c_7^c)$};
		\draw[red] (4-1) -- (4-2) node [pos=0.5, below,black] {\fontsize{8}{1} $(a_2^a,a_8^a)$};
		\draw[red] ($(4-2.north west)+(0.05,-0.068)$) -- ($(3-1.south east)+(-0.05,0.068)$) node [pos=0.7, right,black] {\fontsize{8}{1} $(b_3^\varepsilon,\varepsilon)$};
		
\end{tikzpicture}
	\caption{\label{fig:Verfication-Per}Verification system $V_{1}$.}
\end{figure}
\begin{remark}\upshape
		We conclude this section by  discussing the complexity of checking $\varphi$-diagnosability. 
		First, we note that the augmented system $\tilde{G}$ consists of at most $2|Q|$ states, where $|Q|$ denotes the number of states in system $G$.
		Second, we obtain the observation constrained system $T$ by composing augmented system $\tilde{G}$ with  B\"{u}chi automaton $\mathcal{B}$ that is translated from the LTL formula $\varphi$. 
		Therefore, $T$ has at most $2\cdot|Q|\cdot|X|$ states, 
		where $|X|$ is the number of states in the B\"{u}chi automaton. 
		In general, for any LTL formula $\varphi$, the NBA $\mathcal{B}$ associated with $\varphi$ has at most $ 2^{|\varphi|}\cdot |\varphi|$ number of states, where $|\varphi|$ is the number of operators in formula $\varphi$\cite{baier2008principles}. 
		Then the verification system $V$  has at most $4\cdot|Q|^2 \cdot|X|^2$ states. 
		Finally, checking the existence of the particular cycle in Theorem \ref{thm:verify varphi-diag}  is simply a cycle search problem which can be done in polynomial-time in the size of $V$. 
		Overall, our approach is polynomial in the size of the system model $G$ and exponential in the length of the sensor constraint formula $\varphi$.
\end{remark}

\section{Applications of the Uniform Framework}\label{sec:Application}

Recall that our framework consists of the following four steps:
\begin{enumerate}[1.]
	\item 
	Choose  atomic propositions $\mathcal{AP}$ and labeling function $\textsf{label}: \Sigma_e \to 2^\mathcal{AP}$;
	\item 
	Describe the sensor constraint by  LTL formula $\varphi$;
	\item 
	Construct the observation constrained system $T$ and the verification system $V$;
	\item 
	Check $\varphi$-diagnosability based on $V$ by Theorem \ref{thm:verify varphi-diag}.
\end{enumerate}

In the previous sections, we have shown how to implement Steps 3 and 4. 
However, Steps 1 and 2 are more application-dependent, i.e., one needs to carefully choose $\mathcal{AP}$ and   $\textsf{label}$, 
and write down $\varphi$ based on the specific scenario of the system. 
In this section, we show explicitly  
how the proposed new notion of $\varphi$-diagnosability subsumes existing notions of 
(i) diagnosability under intermittent sensor failures; 
(ii) diagnosability under permanent sensor failures; and
(iii) $K$-loss diagnosability. 
Furthermore, we introduce two new types of diangsability called 
(i) diagnosability under output fairness;  and
(ii) diagnosability under unreliable sensors with minimum dwell-time, using the general notion of $\varphi$-diagnosability.

\subsection{Diagnosability Subject to Sensor Failures}
In Remark~\ref{rem:capture intermittent}, we have discussed how to capture the notion of  robust diagnosability subject to \emph{intermittent sensor failures}  in our framework. 
Furthermore,   we have shown, by the running example, that how to model the notion of robust diagnosability subject to \emph{permanent sensor failures}  in our framework for the special case of a single failure sensor.  
Here we formally present a unified way that supports both intermittent and permanent sensor failures using our framework. 

Let  $G=(Q,\Sigma,\delta,q_0)$ be the system model. 
We assume that  transitions $\mathbb{T}\subseteq Q\times \Sigma$ are partitioned as follows:
\[
\mathbb{T}=\mathbb{T}_{uo}\dot{\cup}\mathbb{T}_{o}\dot{\cup}\mathbb{T}_{\textsf{int}}\dot{\cup}\mathbb{T}_{\textsf{per}},
\]
where 
\begin{itemize}
	\item 
	$\mathbb{T}_{uo}$ are transitions that can never be observed; 
	\item 
	$\mathbb{T}_{o}$ are transitions that can always be observed; 
	\item 
	$\mathbb{T}_{\textsf{int}}$  are transitions subject to intermittent sensor failures;
	\item 
	$\mathbb{T}_{\textsf{per}}$ are transitions subject to permanent sensor failures. 
\end{itemize}
Similar to Equation~\eqref{eq:intermittent observation}, we define an observation mapping $\mathcal{O}:Q\times \Sigma \to 2^{\Sigma\cup \{\varepsilon\}}$ by: 
for any $q \in Q$ and  $\sigma \in \Sigma$, we have
\begin{align} 
	\mathcal{O}(q,\sigma)\!=\!
	\left\{\! 
	\begin{array}{l l}
		\{\varepsilon\} &\text{if } (q,\sigma) \!\in\! \mathbb{T}_{uo} \\
		\{\sigma\} &\text{if } (q,\sigma) \!\in\! \mathbb{T}_{o}\\ 
		\{\sigma,\varepsilon\} &\text{if } (q,\sigma) \!\in\! \mathbb{T}_{\textsf{int}}\cup \mathbb{T}_{\textsf{per}} 
	\end{array} 
	\right..
\end{align}

Note that, for those transitions subject to intermittent sensor failures, there is no need to put additional sensor constraints since they are freely to fail or recover. However, we need to use LTL formula to constrain the behavior of those sensors subject to permanent sensor failures.  
To this end, we choose the following set of  of atomic propositions   
\begin{equation}
\mathcal{AP}_{\textsf{per}}=
\{
m_{i}^{(q,\sigma)}:  i\in \{0,1\}, (q,\sigma)\in \mathbb{T}_{\textsf{per}} 
\},
\end{equation}
where $m_{0}^{(q,\sigma)}$ means that 
the sensor observes event $\sigma$ successfully for transition $(q,\sigma)$, i.e., the corresponding sensor is normal, 
while $m_{1}^{(q,\sigma)}$ means that the sensor misses the observation  for transition $(q,\sigma)$, i.e., the corresponding sensor is faulty. 
Then, the labeling function  
$\textsf{label}_{\textsf{per}}: \Sigma_e \to 2^{\mathcal{AP}_{\textsf{per}}}$ is defined by:
for any $\sigma_e=(q,\sigma,o)\in \Sigma_e$, we have  
\begin{equation}
	\textsf{label}_{\textsf{per}}(\sigma_e)= \left\{
	\begin{array}{cl}
		\{m_0^{(q,\sigma)}\},   & \text{if }  (q,\sigma)\in \mathbb{T}_{\textsf{per}}\wedge   o=\sigma \\
		\{m_1^{(q,\sigma)}\},   & \text{if }  (q,\sigma)\in \mathbb{T}_{\textsf{per}}\wedge  o=\varepsilon\\
		\emptyset,  &  \text{otherwise}   	
	\end{array} \right.
\end{equation}
Then the sensor constraint for the scenario of intermittent/permanent failures can be  written as 
\begin{equation} \label{eq:permanent equation-uniform}
	\varphi_{\textsf{int} \wedge \textsf{per}}=\bigwedge_{(q,\sigma)\in \mathbb{T}_{\textsf{per}}}\Box( m_1^{(q,\sigma)} \to \Box \neg m_0^{(q,\sigma)}).
\end{equation}

\begin{remark}\upshape
	If we only consider the case of intermittent sensor failures, 
	i.e., $\mathbb{T}_{\textsf{per}}=\emptyset$, 
	then the above formulation becomes 
	$\mathcal{AP}_{\textsf{per}}=\emptyset$, 
	$\textsf{label}_{\textsf{per}}(\sigma_e) = \emptyset, \forall \sigma_e\in \Sigma_e$ 
	and $\varphi_{\textsf{int} \wedge \textsf{per}}=\texttt{true}$.  
	Then the B\"{u}chi automaton $\mathcal{B}_{\textsf{int} \wedge \textsf{per}}$ associated with  $\varphi_{\textsf{int} \wedge \textsf{per}}$ only contains a single state with a self-loop labeled with $\texttt{true}$. 
	Therefore, such an NBA $\mathcal{B}_{\textsf{int} \wedge \textsf{per}}$ will not restrict the behavior of $\tilde{G}$ at all, 
	and the verification system $V$ essentially becomes to the standard structures in \cite{carvalho2012robust}. 
\end{remark}

\begin{remark}\upshape
	An approach for the unification of diagnosability of DES subject to both intermittent and permanent sensor failures  has been recently proposed by \cite{takai2021general}. 
	The above presented proposed further generalized the result in \cite{takai2021general} by supporting state-dependent (or transition-based) observations. 
	Furthermore, the result in \cite{takai2021general} is an instance of our general framework, which, as we will further show later, supports much more user-specified scenarios. 
\end{remark}

\subsection{K-loss Robust Diagnosability} \label{subsec:K-loss}
Recently in \cite{oliveira2022k}, the authors propose a new notion of $K$-loss diagnosability in order to capture the scenario of bounded losses in observation channels. 
In this setting\footnote{The results in \cite{oliveira2022k} considers the decentralized setting. Here, we just review its centralized counterpart.}, it is assumed that the event set is partitioned as $\Sigma=\Sigma_o \cup \Sigma_{uo}$. Observable events in $\Sigma_o$ are transmitted from the sensors to the diagnoser via $n\leq |\Sigma_o|$ communication channels. To this end, the set of observable events $\Sigma_o$ is further partitioned as 
\[
\Sigma_o=\Sigma_{o,1} \dot{\cup} \Sigma_{o,2} \dot{\cup} \cdots \dot{\cup}\Sigma_{o,n},
\]
where for each $i\in \{1,\dots,n\}$, 
$\Sigma_{o,i}$ is the set of events whose observations are transmitted via the $i$th communication channel.
Furthermore, it is assumed that for each observation channel $i\in \{1,\dots,n\}$,  
integer  $k_i\in \mathbb{N}$ is the maximum number of \emph{consecutive losses of observations}. 
That is, if event $\sigma\in \Sigma_{o,i}$ occurs $k_i+1$ times consecutively and the diagnoser does not receive its first $k_i$ occurrences due to losses in the observation channel, then its $k_i+1$th occurrence will be received by the diagnoser for sure. 
We denote by  $K=(k_1,\dots,k_n)$ the tuple of all maximum numbers of consecutive loss for all channels. The notion of $K$-loss (co)diagnosability was introduced in \cite{oliveira2022k} as the necessary and sufficient condition for the existence of a diagnoser under such a setting. 
	
Now, we discuss how to formulate $K$-loss  diagnosability  in terms of our general notion of $\varphi$-diagnosability.  For system $G=(Q,\Sigma, f,q_0)$, we define  observation mapping $\mathcal{O}:Q\times \Sigma \to 2^{\Sigma_o\cup \{\varepsilon\}}$ by: 
for any $q \in Q$ and  $\sigma \in \Sigma$, we have
\begin{align} 
	\mathcal{O}(q,\sigma)\!=\!
	\left\{\! 
	\begin{array}{l l}
	\{\varepsilon\} &\text{if } (q,\sigma) \!\in\! \Sigma_{uo} \\ 
	\{\sigma,\varepsilon\} &\text{if } (q,\sigma) \!\in\! \Sigma_{o} 
	\end{array} 
	\right..
\end{align}
Then we choose the set of atomic propositions by 
\begin{equation} \label{eq:AP-K-loss}
	\mathcal{AP}_{\textsf{K-loss}}=
	\{
	m_{i}^{j}:  i\in \{0,1\}, j\in \{1,\dots,n\} 
	\},
\end{equation}
where $m_0^j$ means that an event  $\sigma\in \Sigma_{o,j}$ occurs and it is  transmitted  successfully in the $j$th channel and  $m_1^j$ means that an event transmission in the $j$th channel is lost.  
To capture the above meanings, we define a  labeling function $\textsf{label}_{\textsf{K-loss}}$ by:
for any $\sigma_e=(q,\sigma,o)\in \Sigma_e$, we have  
\begin{equation} \label{eq:label-K-loss}
	\textsf{label}_{\textsf{K-loss}}(\sigma_e)= \left\{
	\begin{array}{cl}
	\{m_0^{j}\},   & \text{if }  \sigma \in \Sigma_{o,j} \wedge   o \neq \varepsilon \\
	\{m_1^{j}\},   & \text{if }  \sigma \in \Sigma_{o,j}\wedge  o=\varepsilon\\
	\emptyset,  &  \text{otherwise}   	
	\end{array} \right.
	\end{equation}
	where $j\in \{1,\dots,n\}$.
	
To formalize the sensor constraint  for $K$-loss diagnosability,  
we need to exclude the scenario, where there are  more than $k_j$  consecutive observation losses for some channel $j\in \{1,\dots, n\}$. To this end, for each $j\in \{1,\dots, n\}$, we define a sequence of LTL formulae $\Phi^j(0),\dots, \Phi^j(k_{j})$ as follows: 
\begin{align}
\left\{
	\begin{array}{l l}
	  \Phi^j(0)= m_1^{j},   \\
	 \Phi^j(k)=m_1^{j} \wedge \bigcirc (\neg m_0^{j} U \Phi^j(k-1)),	
\end{array} \right.
\end{align} 
	where $k\geq 1$. Intuitively,  $\Phi^j(0)=m_1^j$ means that the current observation  in the $j$th channel is lost. 
	Then $\Phi^j(1)=m_1^{j} \wedge \bigcirc (\neg m_0^{j} U m_1^j )$ means that 
	the current observation  in the $j$th channel is lost and 
	the no event $\Sigma_{o,j}$ in the $j$th channel can be observed until the next observation loss in the $j$th channel, which means that 
	the $j$th channel will have two consecutive  losses. 
	Therefore, by induction, $\Phi^j(k)$ means that the observations in the $j$th channel have  $k+1$ consecutive losses. 
	Then, the requirement that ``starting from any instant, the $j$th observation channel cannot have $k_j+1$ consecutive losses" can be captured by the   LTL formula 
	$ \Box \neg \Phi^j(k_j)=\neg \lozenge   \Phi^j(k_j)$. 
	Since we require that all $n$ channels satisfy the $K$-loss assumption, 
	the sensor constraint for this scenario is given by 
\begin{equation}
\varphi_{\textsf{K-loss}}= \bigwedge_{j\in \{1,\dots,n\}} \Box \neg \Phi^j(k_j).
\end{equation}
	
We use following example to illustrate this scenario.

\begin{figure*}[!ht] 
	\centering
	\hspace{-10pt}
	\subfigure[System $G_2$]
	{\label{fig:K-LOSS-G}
		\centering
\usetikzlibrary {patterns.meta}

\begin{tikzpicture}[->,>={Latex}, thick, initial text={}, node distance=1.25cm, initial where=above, thick, base node/.style={circle, draw, minimum size=5mm, font=\small}]		
	\node[initial, state, base node, ] (1) at (1.5,7) {$1$}; 
	\node[state, base node, ](2) at (0.6,6) {$2$};
	\node[state, base node, ](3) [below of=2] {$3$};
	\node[state, base node, ](4) [below of=3] {$4$};
	\node[state, base node, ](5) at (2.4,6) {$5$};
	\node[state, base node, ](6) [below of=5] {$6$};
	\node[state, base node, ](7) [below of=6] {$7$};
	\node[state, base node, ](8) [below of=7] {$8$};
	
	\path[]
	(1) edge node [left, xshift=0.1cm, yshift=0.15cm] {\fontsize{8}{1} $f \backslash\! \{\varepsilon\} $} (2)
	(1) edge node [right, xshift=-0.2cm, yshift=0.15cm] {\fontsize{8}{1} $u \backslash\! \{\varepsilon\}$} (5)
	(2) edge node [right, xshift=-0.1cm] {\fontsize{8}{1} $a\! \backslash\! \{a,\varepsilon\} $} (3)
	(3) edge node [midway,right, xshift=-0.1cm] {\fontsize{8}{1} $b\! \backslash\! \{b,\varepsilon\} $} (4);
	
	\draw (4.west) -- ($(4.west)+(-0.9,0)$)
	($(4.west)+(-0.9,0)$) -- node [right, xshift=-0.2cm] {\fontsize{8}{1} $c\! \backslash\! \{c,\varepsilon\} $} 
	($(2.west)+(-0.9,0)$)
	-- (2.west);
	
	\path[]
	(5) edge node [midway,right, xshift=-0.1cm] {\fontsize{8}{1} $a  \backslash\! \{a,\varepsilon\} $} (6)
	(6) edge node [midway,right, xshift=-0.1cm] {\fontsize{8}{1} $b  \backslash\! \{b,\varepsilon\}$} (7)
	(7) edge node [right, xshift=-0.1cm ] {\fontsize{8}{1} $ c  \backslash\! \{c,\varepsilon\} $} (8);
	
	\draw
	(8.west) -- ($(8.west)+(-0.75,0)$)
	($(8.west)+(-0.75,0)$) -- node [midway,right, xshift=-0.2cm] {\fontsize{8}{1} $a\! \backslash\! \{a,\varepsilon\} $} ($(7.west)+(-0.75,0)$)
	-- (7.west);	
\end{tikzpicture}
	}
\hspace{-10pt}
	\subfigure[B\"{u}chi automaton $\mathcal{B}_{\textsf{K-loss}}$]
	{\label{fig:K-LOSS-B}
		\centering
		\begin{tikzpicture}[->,>=stealth,shorten >=0pt,auto,node distance=2.3cm,thick,base node/.style={circle,draw,minimum size=0.5mm}, font=\small]
	\node[initial,initial text={}, state, accepting,base node] (A) {$A$};
	\node[state, accepting,base node](B)[right of=A ]{$B$};
	
	\path[]
	(A) edge [loop above] node {\fontsize{8}{1}$\neg m_1^1\! \wedge\! \neg m_1^2$} (A)
	(A) edge[bend left] node [below, xshift=0, yshift=-0.1cm] {\fontsize{8}{1}$\neg m_1^1\! \wedge\! m_1^2$} (B)
	(B) edge [loop above] node [above] {\fontsize{8}{1}$\neg m_1^1\! \wedge\! \neg m_1^2\! \wedge\! \neg m_0^2$} (B)
	(B) edge [bend left] node [below] {\fontsize{8}{1}$\neg m_1^1\! \wedge\! \neg m_1^2\! \wedge m_0^2$} (A);
\end{tikzpicture}
	}
\hspace{-10pt}
%
%
	\subfigure[Verification system $V_2$]
	{\label{fig:K-LOSS-V}
		\centering
		\begin{tikzpicture}[->,>={Latex}, thick, initial text={},node distance=1.2cm, initial where=above, thick, base node/.style={rectangle, align = center, draw, minimum height=0mm, rounded corners =1.5mm, font=\small}, ]		
	
	\node[initial, state, base node] (1-2) at (3.9,5) {$\{(1N,A),(1N,A)\}$}; 
	\node[state, accepting,base node](2-2) [below of=1-2] {$\{(3F, A),(6N,A)\}$};
	\node[state,accepting, base node](3-2) [below of=2-2] {$\{(4F,A),(7N,A)\}$};
	\node[state,accepting, base node](4-2) [below of=3-2] {\color{black}$\{(4F,B),(7N,A)\}$};
	\node[state,accepting, base node](5-2) [below of=4-2] {\color{black}$\{(2F,B),(8N,A)\}$};
	
	\node[state,accepting, base node](1-1) at (0,5) {$\{(2F,A),(1N,A)\}$};
	\node[state,accepting, base node](2-1) [below of=1-1] {$\{(2F,A),(5N,A)\}$};
	
	\node[state,accepting, base node](3-1) [below of=2-1] {\color{black}$\{(2F,A),(8N,A)\}$};
	\node[state,accepting, base node](4-1) [below of=3-1] {\color{black}$\{(3F,A),(7N,A)\}$};
	\node[state,accepting, base node,red](5-1) [below of=4-1] {\color{black}$\{(3F,B),(7N,A)\}$};

	\draw (1-2) -- (1-1) node [pos=0.5, above] {\fontsize{8}{1} $({f}_1^\varepsilon, \varepsilon)$};
	\draw[dashed] ($(1-2.south east)+(-0.05,0.068)$) -- +(-45:0.5cm);
	
	\draw (1-1) -- (2-1) node [pos=0.5, left] {\fontsize{8}{1} $(\varepsilon,u_1^\varepsilon)$};
	\draw (2-1) -- (2-2) node [pos=0.5, above] {\fontsize{8}{1} $(a_2^a,a_5^a)$};
	\draw[dashed] ($(2-2.south east)+(-0.05,0.068)$) -- +(-45:0.5cm);
	
	\draw (2-2) -- (3-2) node [pos=0.5, left] {\fontsize{8}{1} $(b_3^b,b_6^b)$};
	\draw (3-2) -- (3-1) node [pos=0.5, above] {\fontsize{8}{1} $(c_4^c,c_7^c)$};
	\draw (3-1) -- (4-1) node [pos=0.5, left,black] {\fontsize{8}{1} $(a_2^a,a_8^a)$};
	\draw (4-1) -- (4-2) node [pos=0.5, above,black] {\fontsize{8}{1} $(b_3^\varepsilon,\varepsilon)$};
	\draw (4-2) -- (5-2) node [pos=0.5, left,black] {\fontsize{8}{1} $(c_4^c,c_7^c)$};
	\draw (5-2) -- (5-1) node [pos=0.5, above,black] {\fontsize{8}{1} $(a_2^a,a_8^a)$};
	
\end{tikzpicture}
	}
\hspace{-10pt}
	\caption{An example of K-loss robust diagnosability. For system $G_2$: $\Sigma_{o}=\Sigma_{o,1}\cup \Sigma_{o,2}, \Sigma_{o,1}=\{a,c\}$ and $ \Sigma_{o,2}=\{b\}$.}
	\end{figure*}
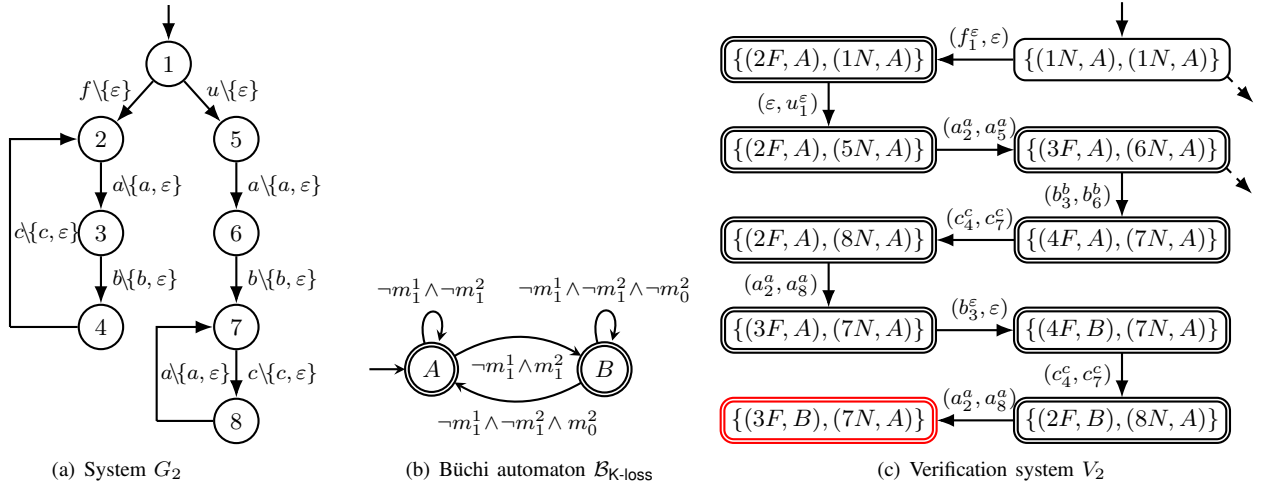
				
\begin{myexm} \label{exm:K-loss example}\upshape
	We consider the system $G_2$, shown in Figure~\ref{fig:K-LOSS-G}, where the observable events $\Sigma_o=\{a,b,c\}$ are partitioned into two observation channels with $\Sigma_{o,1}=\{a,c\}$ and $\Sigma_{o,2}=\{b\}$. 
	We assume that $k_1=0$ and $k_2=1$, i.e., the first channel for  $\Sigma_{o,1}$ is reliable without observation loss and the second channel for $\Sigma_{o,2}$ can only have at most one consecutive observation loss. 
	Then the set atomic propositions is $\mathcal{AP}_{\textsf{K-loss}}=\{m_0^1,m_1^1, m_0^2, m_1^2\}$, the labeling function is specified in Equation~\eqref{eq:label-K-loss} and the sensor constraint is 
	\[
		\varphi_{\textsf{K-loss}}= \Box \neg m_1^1	\wedge \Box \neg (m_1^{2} \wedge \bigcirc (\neg m_0^{2} U m_1^{2} )).
	\]
	The B\"{u}chi automaton $\mathcal{B}_{\textsf{K-loss}}$ associated with $\varphi_{\textsf{K-loss}}$ is shown in Figure~\ref{fig:K-LOSS-B}. 
	We construct the verification system $V_2$, which is partially depicted in Figure~\ref{fig:K-LOSS-V}. 
Due to the sensor constraint $\varphi_{\textsf{K-loss}}$, if the event $b^{\varepsilon}_3$ or $b^{\varepsilon}_6$ occurs in $G_2$, the next occurrence of $b^{\varepsilon}_3$ or $b^{\varepsilon}_6$ is possible only after $b^{b}_3$ or $b^{b}_6$ occurs. 
	For example, from state $3$, if the event $b^{\varepsilon}_3$ occurs, we can only obtain the extended string $b^{\varepsilon}_3 c_4^c a_2^a b_3^b$ rather than $b^{\varepsilon}_3 c_4^c a_2^a b_3^\varepsilon$ since the latter corresponds to the case of two consecutive losses in channel $\Sigma_{o,2}$.
	Furthermore, for the first communication channel, since its maximum failure bound is $k_1=0$, that is, it is reliable, the extended events such as $a_2^\varepsilon$ and $c_4^\varepsilon$ will not be feasible. 
	In Figure~\ref{fig:K-LOSS-V}, we see that the system is blocked at state $\{(3F,B),(7N,A)\}$. This is because that the first component $(3F,B)$ means that the second communication channel has lost an observation and, therefore,  has to transit next observation successfully. 
	However, only event $b_3^b$ can occur from state $(3F,B)$, while only event $c_7^c$ can occur at state $(7N,A)$.
	With the similar reason, one can also complete the remaining part of $V_2$
	in which no reachable cycle  satisfying the conditions in Theorem \ref{thm:verify varphi-diag} can be found. 
	Therefore,  system $G_2$ is $\varphi_{\textsf{K-loss}}$-diagnosable.  
\end{myexm}

\subsection{Unreliable Sensors with Minimum Dwell-Time}
We have shown that both intermittent sensor failures and permanent sensor failures can be captured in our unified framework. Note that, in the existing framework of intermittent sensor failures, those unreliable sensors can fail and recover \emph{arbitrarily} at any instant.  
Here, we show how our LTL-based approach can describe a more general scenario, 
where sensors may switch between failure mode and normal mode but with \emph{minimum dwell-time} for each mode. This is motivated by the practical scenario that
whenever a normal sensor breaks down, it will take some time to recover, 
and whenever a sensor failure is fixed, it can ensure to  work normally for some periods.   

Still, we consider system $G=(Q,\Sigma,\delta,q_0)$ and  assume that the event set is partitioned as $\Sigma=\Sigma_o \dot{\cup} \Sigma_{uo}\dot{\cup} \Sigma_{ur}$, 
where $\Sigma_{ur}$ is the set of events whose sensors are unreliable. 
We  denote by $k_N\in \mathbb{N}$ the minimum dwell-time for the normal mode 
(or simply, \emph{normal dwell-time}), 
which means that whenever  a sensor recovers from failure to normal, it will remain normal at least for the next $k_N$ uses of the sensor.  
Similarly, we  denote by $k_F\in \mathbb{N}$ the minimum dwell-time for the failure mode 
(or simply, \emph{failure dwell-time}), 
which means that whenever  a sensor breaks down from normal to failure, it will remain failure at least for the next $k_F$ uses of the sensor. 
%

To formalized the above setting in our framework, we define an observation mapping $\mathcal{O}:Q\times \Sigma \to 2^{\Sigma_o\cup \Sigma_{ur}\cup \{\varepsilon\}}$ by: 
for any $q \in Q$ and  $\sigma \in \Sigma$, we have
\begin{align} 
	\mathcal{O}(q,\sigma)\!=\!
	\left\{\! 
	\begin{array}{l l}
		\{\varepsilon\} &\text{if }  \sigma \!\in\! \Sigma_{uo} \\
		\{\sigma\} &\text{if } \sigma \!\in\! \Sigma_{o}\\ 
		\{\sigma,\varepsilon\} &\text{if } \sigma \!\in\! \Sigma_{ur}
	\end{array} 
	\right..
\end{align}
The set of atomic propositions is chosen as  
\begin{equation}
	\mathcal{AP}_{\textsf{dwell}}=\{ m_i^{\sigma}: i=0,1, \sigma\in \Sigma_{ur} \},   
\end{equation}
where $m_0^\sigma$ means that event $\sigma\in \Sigma_{ur}$ occurs and its sensor is normal, while $m_1^\sigma$ means that event $\sigma\in \Sigma_{ur}$ occurs and its sensor is failure. The  labeling function $\textsf{label}_{\textsf{dwell}}$ is defined by:
for any $\sigma_e=(q,\sigma,o)\in \Sigma_e$, we have  
\begin{equation} \label{eq:label-dwell}
	\textsf{label}_{\textsf{dwell}}(\sigma_e)= \left\{
	\begin{array}{cl}
		\{m_0^{\sigma}\},   & \text{if }   \sigma \in \Sigma_{ur}\wedge   o \neq \varepsilon \\
		\{m_1^{\sigma}\},   & \text{if }  \sigma \in \Sigma_{ur}\wedge  o=\varepsilon\\
		\emptyset,  &  \text{otherwise}   	
	\end{array} \right..
\end{equation}

To capture the sensor constraint of normal dwell-time, 
we first need to exclude the scenarios where a sensor switches from the failure mode to the normal mode and  remains normal for less than $k_N$ uses before the next switch back to the failure mode. 
To this end, we define a sequence of LTL formulae inductively as follows:
\begin{align}
\left\{
	\begin{array}{l l}
	  \Phi^{\sigma}_{\textsf{nor}}(1)= m_0^{\sigma} \wedge \bigcirc (\neg m_0^{\sigma} U m_1^{\sigma})   \\
	\Phi^{\sigma}_{\textsf{nor}}(k)= m_0^{\sigma} \wedge \bigcirc ((\neg m_0^{\sigma} \wedge \neg m_1^{\sigma}) U \Phi^{\sigma}_{\textsf{nor}}(k-1) )
\end{array} \right.\!\!\!.
\end{align} 
Intuitively, 
$\Phi^{\sigma}_{\textsf{nor}}(1)$ means that,  $\sigma$ is observed currently with a normal sensor,  and there is  no occurrence  of  $\sigma$ until its sensor fails. 
Therefore, this formula captures the scenario where the sensor of event $\sigma$ remains normal for one time of use before switching to the failure mode. 
By induction, $\Phi^{\sigma}_{\textsf{nor}}(k)$ means that the sensor of event $\sigma$ has  remained normal for  $k$ times of use before switching to the failure mode. 

Then, since the dwell-time of normal mode is counted starting from the transition from failure mode to normal mode,   
we can describe the scenario where sensor switches from failure mode to normal mode and remains at normal mode for $k$ uses before switching back to failure mode by the following formula: 
\begin{equation} \label{eq:dwell-pro-n}
	\Phi^{\sigma}_{\textsf{nor},\textsf{dwell}}(k)=m_1^\sigma \wedge \bigcirc (\neg m_0^\sigma U \Phi^{\sigma}_{\textsf{nor}}(k) ).
\end{equation}
However, the minimum normal dwell-time $k_N$ requires that  the above formula should not be satisfied for any $k\leq k_N-1$; otherwise, it means that the sensor fails again when it switches back to normal for less than $k_N$ times of uses. 
Therefore, the overall, normal dwell-time constraint is given by 
\begin{equation}
\Phi_{\textsf{nor},\textsf{dwell}}  =\bigwedge_{\sigma \in \Sigma_{ur}} \bigwedge_{k=1,\dots, k_N-1}
  \Box \neg \Phi^{\sigma}_{\textsf{nor},\textsf{dwell}}(k). 
\end{equation}

Similar to the case of normal dwell-time, we can define the sensor constraint for failure dwell-time by 
\begin{equation}
\Phi_{\textsf{fail},\textsf{dwell}}  =\bigwedge_{\sigma \in \Sigma_{ur}} \bigwedge_{k=1,\dots, k_F-1}
  \Box \neg \Phi^{\sigma}_{\textsf{fail},\textsf{dwell}}(k),  
\end{equation}
where we have 
\begin{align}
\left\{
	\begin{array}{l l}
	\Phi^{\sigma}_{\textsf{fail},\textsf{dwell}}(k)=m_0^\sigma \wedge \bigcirc (\neg m_1^\sigma U \Phi^{\sigma}_{\textsf{fail}}(k))
	\\
\Phi^{\sigma}_{\textsf{fail}}(1)=m_1^{\sigma} \wedge \bigcirc (\neg m_1^{\sigma} U m_0^{\sigma})   
\\
	\Phi^{\sigma}_{\textsf{fail}}(k)=m_1^{\sigma} \wedge \bigcirc ((\neg m_1^{\sigma} \wedge \neg m_0^{\sigma}) U \Phi^{\sigma}_{\textsf{fail}}(k-1)
\end{array} \right.\!\!\!.
\end{align} 

Therefore, the overall sensor constraint for both  normal and failure dwell-times is given by  
\begin{equation}\label{eq:dwell-pro}
	\varphi_{\textsf{dwell}}=\Phi_{\textsf{nor},\textsf{dwell}}\wedge \Phi_{\textsf{fail},\textsf{dwell}}.
\end{equation}

\subsection{Diagnosability with Output Fairness}\label{subsec:fair}
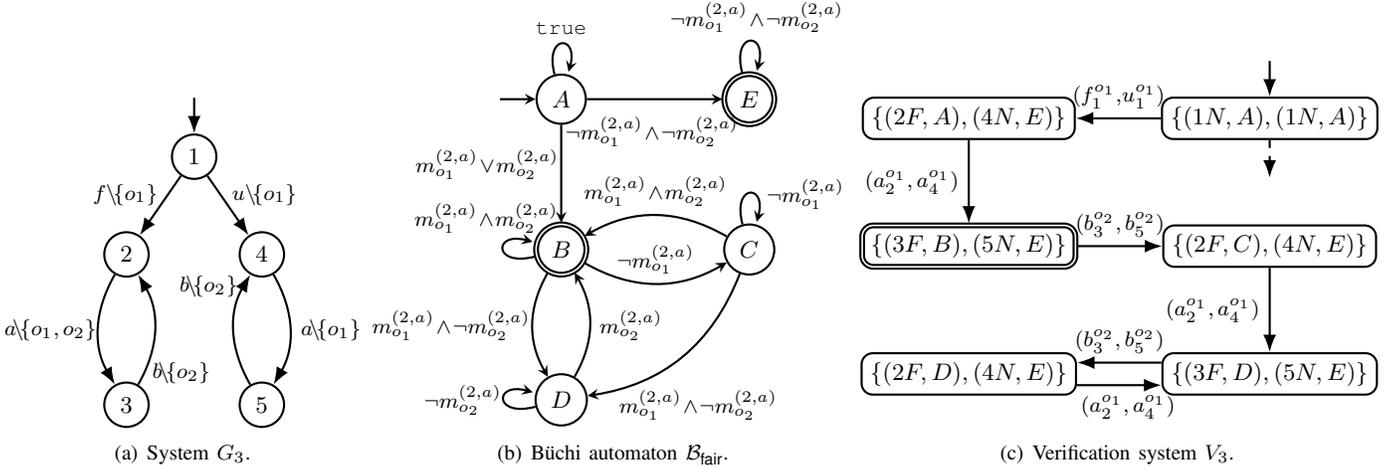
\begin{figure*}[!ht] 
	\centering
	\hspace{-20pt}
	\subfigure[System $G_3$.]
	{\label{fig:G-fair}
		\centering
		\begin{tikzpicture}[->,>={Latex}, thick, initial text={}, node distance=2cm, initial where=above, thick, base node/.style={circle, draw, minimum size=5mm, font=\small}]		
	\node[initial, state, base node, ] (1) at (1.5,6) {$1$}; 
	\node[state, base node, ](2) at (0.6,4.7) {$2$};
	\node[state, base node, ](3) [below of=2] {$3$};
	\node[state, base node, ](4) at (2.4,4.7) {$4$};
	\node[state, base node, ](5) [below of=4] {$5$};
	
	\path[]
	(1) edge node [left, xshift=0.1cm, yshift=0.15cm] {\fontsize{8}{1} $f \backslash\! \{o_1\} $} (2)
	(1) edge node [right, xshift=-0.2cm, yshift=0.15cm] {\fontsize{8}{1} $u \backslash\! \{o_1\}$} (4)
	(2) edge[bend right] node [left, xshift=0.1cm] {\fontsize{8}{1} $a\! \backslash\! \{o_1,o_2\} $} (3)
	(3) edge[bend right] node [below, xshift=0.3cm, yshift=-0.3cm] {\fontsize{8}{1} $b\! \backslash\! \{o_2\} $} (2)
	(4) edge[bend left] node [right, xshift=-0.1cm] {\fontsize{8}{1} $a\! \backslash\! \{o_1\} $} (5)
	(5) edge[bend left] node [above, xshift=-0.4cm, yshift=0.3cm] {\fontsize{8}{1} $b\! \backslash\! \{o_2\} $} (4);	
\end{tikzpicture}
	}
	\hspace{-20pt}
	\subfigure[B\"{u}chi automaton $\mathcal{B}_{\textsf{fair}}$.]
	{\label{fig:NBA-fair}
		\centering
		\begin{tikzpicture}[->,>={stealth},shorten >=0pt,auto,node distance=2.5cm,
	thick,base node/.style={circle,draw,minimum size=0.5mm}, font=\small]
	\node[initial,initial text={}, state,base node] (A) at (0,5) {$A$};
	
	\node[state, accepting,base node](B) at (0,3){$B$};
	\node[state, base node](D)at (0,1){$D$};
	\node[state, base node](C)[right of=B ]{$C$};
	\node[state, accepting,base node](E)[right of=A ]{$E$};
	
	\path[]
	(A) edge [loop above] node {\fontsize{8}{1}$\texttt{true}$} (A)
	(A) edge node [pos=0.4,left, xshift=0.1cm] {\fontsize{8}{1} $m_{o_1}^{(2,a)}\! \vee\! m_{o_2}^{(2,a)}$} (B)
	(A) edge node [pos=0.45,below, yshift=-0.1cm] {\fontsize{8}{1} $\neg m_{o_1}^{(2,a)}\! \wedge\! \neg m_{o_2}^{(2,a)}$} (E)
	(B) edge [loop left] node[above,xshift=-0.23cm,yshift=0.1cm]{\fontsize{8}{1}$m_{o_1}^{(2,a)}\! \wedge\! m_{o_2}^{(2,a)}$} (B)
	(D) edge[bend right] node [right] {\fontsize{8}{1}$m_{o_2}^{(2,a)}$} (B)
	(B) edge[bend right] node [left, xshift=0.1cm] {\fontsize{8}{1}$m_{o_1}^{(2,a)}\! \wedge\! \neg m_{o_2}^{(2,a)}$} (D)
	(B) edge[bend right] node [pos=0.5, above] {\fontsize{8}{1}$\neg m_{o_1}^{(2,a)}$} (C)
	(C) edge[bend right] node [pos=0.5, above] {\fontsize{8}{1}$m_{o_1}^{(2,a)}\! \wedge\! m_{o_2}^{(2,a)}$} (B)
	(C) edge[bend left] node [pos=0.5, below, yshift=-0.3cm, xshift=0.22cm] {\fontsize{8}{1}$m_{o_1}^{(2,a)}\! \wedge\! \neg m_{o_2}^{(2,a)}$} (D)
	(C) edge [loop above] node[right,xshift=0.1cm]{\fontsize{8}{1}$\neg m_{o_1}^{(2,a)}$} (C)
	(D) edge [loop left] node[left,xshift=0.1cm]{\fontsize{8}{1}$\neg m_{o_2}^{(2,a)}$} (D)
	(E) edge [loop above] node[above,xshift=0.1cm]{\fontsize{8}{1}$\neg m_{o_1}^{(2,a)}\! \wedge\! \neg m_{o_2}^{(2,a)}$} (E);
\end{tikzpicture}
	}
	\hspace{-20pt}
	\subfigure[Verification system $V_3$.]
	{\label{fig:V-fair}
		\centering
		\begin{tikzpicture}[->,>={Latex}, thick, initial text={},node distance=1.7cm, initial where=above, thick, base node/.style={rectangle, align = center, draw, minimum height=0mm, rounded corners =1.5mm, font=\small}, ]		
	
	\node[initial, state, base node] (1-2) at (4,4) {$\{(1N,A),(1N,A)\}$}; 
	\node[state, base node](2-2) [below of=1-2] {$\{(2F, C),(4N,E)\}$};
	\node[state,base node](3-2) [below of=2-2] {$\{(3F,D),(5N,E)\}$};
	
	\node[state,base node](1-1) at (0,4) {$\{(2F,A),(4N,E)\}$};
	\node[state,accepting, base node](2-1) [below of=1-1] {$\{(3F,B),(5N,E)\}$};
	
	\node[state, base node](3-1) [below of=2-1] {\color{black}$\{(2F,D),(4N,E)\}$};

	\draw (1-2) -- (1-1) node [pos=0.55, above] {\fontsize{8}{1} $({f}_1^{o_1},\! {u}_1^{o_1})$};
	\draw[dashed] ($(1-2.south)$) -- +(-90:0.5cm);
	\draw (1-1) -- (2-1) node [pos=0.5, left] {\fontsize{8}{1} $(a_2^{o_1},a_4^{o_1})$};
	\draw (2-1) -- (2-2) node [pos=0.45, above] {\fontsize{8}{1} $(b_3^{o_2},b_5^{o_2})$};
	\draw (2-2) -- (3-2) node [pos=0.5, left] {\fontsize{8}{1} $(a_2^{o_1},a_4^{o_1})$};
	%
	%
	\path[]
	($(3-2.west)+(0,0.15)$) edge node [above, pos=0.55] {\fontsize{8}{1} $ (b_3^{o_2},b_5^{o_2})$} ($(3-1.east)+(0,0.15)$)
	($(3-1.east)-(0,0.15)$) edge node [below, pos=0.55] {\fontsize{8}{1} $ (a_2^{o_1},a_4^{o_1})$} ($(3-2.west)-(0,0.15)$);	
\end{tikzpicture}
	}
	\hspace{-20pt}
	\caption{An example of diagnosability with output fairness. For system $G_3$: $\Sigma=\{a,b,f,u\}$ and $ \Delta=\{o_1,o_2\}$.  }
\end{figure*}

In the previous three subsections, we have $\Delta\subseteq \Sigma$ in the observation mapping  due to the specific  meanings of sensor failures or observation losses. 
Here, we further consider the general case where observation symbols $\Delta$
and internal events $\Sigma$ are different. 

As a motivating example, let  us consider  system $G_3$  shown in Figure~\ref{fig:G-fair} 
with  $\Sigma=\{a,b, f,u\}, \Sigma_F=\{f\}$ and $\Delta=\{o_1,o_2\}$. 
The output function $\mathcal{O}:Q\times \Sigma \to 2^{\Delta\cup \{\varepsilon\}}$ is specified in Figure~\ref{fig:G-fair}. 
Note that, without any sensor constraint, the system is not diagnosable according to Definition~\ref{def:classical diag} as defined in   \cite{takai2012verification}.  
This is because one can observe output sequence $o_1(o_1o_2)^n$ for an arbitrary large $n$ for which we cannot determine whether the underlying internal string is faulty or normal.

However, this above issue of non-diagnosability is not very practical since it is \emph{unfair} for output $o_2$ at transition $(2,a)$. Specifically, the reason why we may observe output sequence $o_1(o_1o_2)^\omega$ is that, after the occurrence of fault event $f$, every time the system fires transition $(2,a)$, output symbol $o_2$ losses its opportunity to be observed. 
However, if we assume that the occurrences of $o_1$ and $o_2$ are purely random with some (unknown but non-zero) probabilities, then $o_2$ will be observed with probability one at the system goes on and upon its occurrence, we can detect the fault. 

Instead of going to a probability framework to resolve the above issue \cite{thorsley2008diagnosability,athanasopoulou2010maximum}, 
here we propose to use \emph{fairness condition}, which is a widely used   assumption   in reactive systems, to rule out those  unrealistic infinite behaviors of  the sensors. 
Formally, still let  $G=(Q,\Sigma,\delta,q_0)$ be the system model with an arbitrary output function $\mathcal{O}:Q\times \Sigma \to 2^{\Delta \cup \{\varepsilon\}}$.  
We assume that 
\[
\mathbb{T}_{\textsf{fair}}\subseteq Q\times \Sigma
\]
is the set of transitions for which all of its outputs are \emph{fair} in the sense that 
if  transition $(q, \sigma)\in \mathbb{T}_{\textsf{fair}}$ is fired for infinite number of times, then all of its  observations $o\in \mathcal{O}(q,\sigma)$ can be observed for infinite number of times.   
To formalize this requirement, we choose the set of atomic propositions as follows:
\begin{equation}
\mathcal{AP}_{\textsf{fair}}=
\{
m_{o}^{(q,\sigma)}:  (q,\sigma)\in \mathbb{T}_{\textsf{fair}}, o \in \mathcal{O}(q,\sigma) 
\},
\end{equation}
where $m_{o}^{(q,\sigma)}$ means   that the output of transition $(q,\sigma)\in \mathbb{T}_{\textsf{fair}}$ is $o \in \mathcal{O}(q,\sigma)$.  
For brevity, for each transition 
$(q, \sigma)\in \mathbb{T}_{\textsf{fair}}$, 
we denote by $m^{(q,\sigma)}$ the proposition that transition $(q,\sigma)$ occurs, i.e., 
$m^{(q,\sigma)}=\bigvee_{o \in \mathcal{O}(q,\sigma)} m_o^{(q,\sigma)}$

Based on $\mathcal{AP}_{\textsf{fair}}$, we define the labeling function by: for any $\sigma_e=(q,\sigma,o)\in \Sigma_e$, we have  
\begin{equation}
	\textsf{label}_{\textsf{fair}}(\sigma_e)= \left\{
	\begin{array}{cl}
		\{m_o^{(q,\sigma)}\},   & \text{if }  (q,\sigma)\in \mathbb{T}_{\textsf{fair}} \\
		\emptyset,  &  \text{otherwise}   	
	\end{array} \right..
\end{equation}
Then the sensor constraint for the scenario of fairness assumption can be written as 
\begin{equation} \label{eq:fairness}
	\varphi_{\textsf{fair}}\!:=\!
	\bigwedge_{(q,\sigma) \in \mathbb{T}_{\textsf{fair}}} \!
	\left(\! \Box \lozenge   m^{(q,\sigma)}  \to\! \bigwedge_{o \in \mathcal{O}(\sigma,q)}  \Box \lozenge m_o^{(q,\sigma)}\! \right).
\end{equation}

We use following example to illustrate this scenario.
\begin{myexm}\upshape
Still we  consider  system $G_3$, which is shown to be non-diagnosable according to Definition~\ref{def:classical diag}. 
Now we show that the system is actually $\varphi$-diagnosable under the fairness assumption for outputs.
Specifically,  we define $\mathbb{T}_{\textsf{fair}}=\{(2,a)\}$ as the set of fair transition. 
Then the set of atomic propositions is  $\mathcal{AP}_{\textsf{fair}}=\{m_{o_1}^{(2,a)},m_{o_2}^{(2,a)}\}$.
Then the sensor constraint for output fairness can be written as 
	\begin{align*}
		&\varphi_{\textsf{fair}}=\\
		&\left( \Box \lozenge (m_{o_1}^{(2,a)} \vee m_{o_2}^{(2,a)}) \right) \to  \left((\Box \lozenge m_o^{(q,\sigma)}) \wedge (\Box \lozenge m_o^{(q,\sigma)}) \right),
	\end{align*}
which means that if  transition $(2,a)$ occurs infinitely, then we will observe each possible output in $\mathcal{O}(2,a)=\{o_1,o_2\}$ for infinite number of times. 
The B\"{u}chi automaton $\mathcal{B}_{\textsf{fair}}$ associated with LTL formula $\varphi_{\textsf{fair}}$ is shown in Figure~\ref{fig:NBA-fair}.  
It is worth noticing that, in contrast to B\"{u}chi automata in the previous examples, where all states are accepting, here only some states in $\mathcal{B}_{\textsf{fair}}$ are accepting. 
This is because that  fairness condition  $\varphi_{\textsf{fair}}$ is an \emph{liveness} type of property, while all previous conditions $	\varphi_{\textsf{int} \wedge \textsf{per}}, \varphi_{\textsf{K-loss}}$ and $\varphi_{\textsf{dwell}}$ belong to the category of \emph{safety} properties. 
To verify $\varphi_{\textsf{fair}}$-diagnosability, we construct the verification system $V_3$ of system $G_3$, which is partially depicted in Figure~\ref{fig:NBA-fair}. 
Note that, although there is a cycle$\{(3F,D),(5N,E)\}\xrightarrow{(b_3^{o_2},b_5^{o_2})} \{(2F,D),(4N,E)\}\xrightarrow{(a_2^{o_1},a_4^{o_1})} \{(3F,D),(5N,E)\}$ in $V_3$, it does not satisfy the conditions in Theorem \ref{thm:verify varphi-diag} since the states in the cycle are not accepting. Essentially, this cycle corresponds to the unfair extended string making the system not diagnosable without fairness assumption. 
One can also check that, in the remaining part of  $V_3$, there is no cycle  satisfying the conditions in Theorem \ref{thm:verify varphi-diag}. That is, system $G_3$ is $\varphi_{\textsf{fair}}$-diagnosable.
\end{myexm}

\begin{remark}\upshape
We note that the verification of diagnosability of fair DES has been invesitigated by 
\cite{biswas2010fairness,germanos2015diagnosability, bittner2022diagnosability}.
However, the notion of fairness is different from our setting here. 
Specifically, existing works assume that the dynamic of the system is fair in the sense that 
each transition can be executed for infinite number of times whenever it is enabled infinitely.
However, the observation mappings considered therein are still statically modeled as a natural projection.	In contrast, we impose fairness constraint on the observation mapping rather than the internal behavior of the system  the system's dynamic.   
Furthermore, diagnosability with output fairness has been studied in  our preliminary work \cite{dong2022fault}. However, the verification algorithm in \cite{dong2022fault} is customized only to the case of output fairness. Here we address this problem within our general framework using the unified verification procedure.  
\end{remark}
				
\subsection{Diagnosability   with Mixed Sensor Constraints}\label{subsec:mixed consitraint}
In the previous subsections, we have shown explicitly how different types of sensor constraints can be formulated within our general framework. The user is expected to write down more different sensor constraints for different applications. 
Here, we would like to further point out that our framework is flexible to support  DES where different sensors belong to different types of constraints. 

For example, one can  classify the sensors into $n$ different types and select appropriate atomic proposition set $\mathcal{AP}_i$, labeling function $\textsf{label}_{i}$ and  sensor constraint  $\varphi_i$ for each type of sensors. Then the overall sensor constraint is simply given by $\varphi_{\textsf{mix}}= \bigwedge_{i=1}^{n} \varphi_i$. 
As long as the behaviors of the sensors are not coupled with each other, $\varphi_{\textsf{mix}}$ is always well-defined without conflicting among each sub-formula $\varphi_i$.

\section{Conclusion}\label{sec:Conclusion}
In this paper,  we developed a uniform framework for diagnosability analysis of DES subject to unreliable sensors. To this end, we used LTL formulae as a general and user-friendly tool to model unreliable sensors without restricting to specific sensor types.  We proposed a new notion of $\varphi$-diagnosability as well as its effective verification procedure.  Our approach leverages the automated transformation from LTL formulae to B\"{u}chi automata, which avoids to hand-code the sensor automata case-by-case.  
Our framework not only unifies different  existing  notions of robust diagnosability of DES subject to unreliable sensors, but also supports new scenarios of unreliable sensors for which the diagnosability verification problems have never been considered. In the future, we would like to extend our framework to the decentralized setting with multiple local diagnosers. Also, we would like to further investigate the active diagnosis problem for DES subject to unreliable sensors within our uniform framework.

\bibliographystyle{plain}
	\bibliography{des}

\begin{thebibliography}{10}

\bibitem{alves2021robust}
M.~Alves, A.~Cunha, L.~Carvalho, M.~Moreira, and J.~Basilio.
\newblock Robust supervisory control of discrete event systems against
  intermittent loss of observations.
\newblock {\em International Journal of Control}, 94(7):2008--2020, 2021.

\bibitem{athanasopoulou2010maximum}
E.~Athanasopoulou, L.~Li, and C.~Hadjicostis.
\newblock Maximum likelihood failure diagnosis in finite state machines under
  unreliable observations.
\newblock {\em IEEE Transactions on Automatic Control}, 55(3):579--593, 2010.

\bibitem{baier2008principles}
C.~Baier and J.~Katoen.
\newblock {\em Principles of Model Checking}.
\newblock MIT press, 2008.

\bibitem{basile2017diagnosability}
F.~Basile, M.P. Cabasino, and C.~Seatzu.
\newblock Diagnosability analysis of labeled time {P}etri net systems.
\newblock {\em IEEE Transactions on Automatic Control}, 62(3):1384--1396, 2017.

\bibitem{basilio2021analysis}
J.~Basilio, C.~Hadjicostis, and R.~Su.
\newblock Analysis and control for resilience of discrete event systems: Fault
  diagnosis, opacity and cyber security.
\newblock {\em Foundations and Trends{\textregistered} in Systems and Control},
  8(4):285--443, 2021.

\bibitem{biswas2010fairness}
S.~Biswas, D.~Sarkar, S.~Mukhopadhyay, and A.~Patra.
\newblock Fairness of transitions in diagnosability of discrete event systems.
\newblock {\em Discrete Event Dynamic Systems}, 20(3):349--376, 2010.

\bibitem{bittner2022diagnosability}
B.~Bittner, M.~Bozzano, A.~Cimatti, M.~Gario, S.~Tonetta, and V.~Vozarova.
\newblock Diagnosability of fair transition systems.
\newblock {\em Artificial Intelligence}, page 103725, 2022.

\bibitem{boussif2021intermittent}
A.~Boussif, M.~Ghazel, and J.~Basilio.
\newblock Intermittent fault diagnosability of discrete event systems: an
  overview of automaton-based approaches.
\newblock {\em Discrete Event Dynamic Systems}, 31(1):59--102, 2021.

\bibitem{cao2021weak}
L.~Cao, S.~Shu, F.~Lin, Q.~Chen, and C.~Liu.
\newblock Weak diagnosability of discrete event systems.
\newblock {\em IEEE Transactions on Control of Network Systems}, 2022.

\bibitem{carvalho2012robust}
L.~Carvalho, J.~Basilio, and M.~Moreira.
\newblock Robust diagnosis of discrete event systems against intermittent loss
  of observations.
\newblock {\em Automatica}, 48(9):2068--2078, 2012.

\bibitem{carvalho2011generalized}
L.~Carvalho, M.~Moreira, and J.~Basilio.
\newblock Generalized robust diagnosability of discrete event systems.
\newblock {\em IFAC Proceedings Volumes}, 44(1):8737--8742, 2011.

\bibitem{carvalho2021comparative}
L.~Carvalho, M.~Moreira, and J.~Basilio.
\newblock Comparative analysis of related notions of robust diagnosability of
  discrete-event systems.
\newblock {\em Annual Reviews in Control}, 2021.

\bibitem{carvalho2013robust}
L.~Carvalho, M.~Moreira, J.~Basilio, and S.~Lafortune.
\newblock Robust diagnosis of discrete-event systems against permanent loss of
  observations.
\newblock {\em Automatica}, 49(1):223--231, 2013.

\bibitem{cassandras2008introduction}
C.~Cassandras and S.~Lafortune.
\newblock {\em Introduction to Discrete Event Systems}, volume~3.
\newblock Springer, 2021.

\bibitem{cassez2012complexity}
F.~Cassez.
\newblock The complexity of codiagnosability for discrete event and timed
  systems.
\newblock {\em IEEE Transactions on Automatic Control}, 57(7):1752--1764, 2012.

\bibitem{chen2018revised}
J.~Chen, C.~Keroglou, C.N. Hadjicostis, and R.~Kumar.
\newblock Revised test for stochastic diagnosability of discrete-event systems.
\newblock {\em IEEE Transactions on Automation Science and Engineering},
  15(1):404--408, 2018.

\bibitem{chen2015fault}
J.~Chen and R.~Kumar.
\newblock Fault detection of discrete-time stochastic systems subject to
  temporal logic correctness requirements.
\newblock {\em IEEE Transactions on Automation Science and Engineering},
  12(4):1369--1379, 2015.

\bibitem{dong2022fault}
W.~Dong, S.~Gao, X.~Yin, and S.~Li.
\newblock Fault diagnosis of discrete-event systems under non-deterministic
  observations with output fairness.
\newblock In {\em 61th IEEE Conference on Decision and Control (CDC)}, 2022.
\newblock submitted, available at \url{https://arxiv.org/pdf/2204.02617}.

\bibitem{germanos2015diagnosability}
V.~Germanos, S.~Haar, V.~Khomenko, and S.~Schwoon.
\newblock Diagnosability under weak fairness.
\newblock {\em ACM Transactions on Embedded Computing Systems}, 14(4):1--19,
  2015.

\bibitem{giannakopoulou2002states}
D.~Giannakopoulou and F.~Lerda.
\newblock From states to transitions: Improving translation of {LTL} formulae
  to {B{\"u}chi} automata.
\newblock In {\em International Conference on Formal Techniques for Networked
  and Distributed Systems}, pages 308--326. Springer, 2002.

\bibitem{hu2020design}
Y.~Hu, Z.~Ma, and Z.~Li.
\newblock Design of supervisors for active diagnosis in discrete event systems.
\newblock {\em IEEE Transactions on Automatic Control}, 65(12):5159--5172,
  2020.

\bibitem{hu2021diagnosability}
Y.~Hu, Z.~Ma, Z.~Li, and A.~Giua.
\newblock Diagnosability enforcement in labeled {P}etri nets using supervisory
  control.
\newblock {\em Automatica}, 131:109776, 2021.

\bibitem{jiang2004failure}
S.~Jiang and R.~Kumar.
\newblock Failure diagnosis of discrete-event systems with linear-time temporal
  logic specifications.
\newblock {\em IEEE Transactions on Automatic Control}, 49(6):934--945, 2004.

\bibitem{kanagawa2015diagnosability}
N.~Kanagawa and S.~Takai.
\newblock Diagnosability of discrete event systems subject to permanent sensor
  failures.
\newblock {\em International Journal of Control}, 88(12):2598--2610, 2015.

\bibitem{lafortune2018history}
S.~Lafortune, F.~Lin, and C.~Hadjicostis.
\newblock On the history of diagnosability and opacity in discrete event
  systems.
\newblock {\em Annual Reviews in Control}, 45:257--266, 2018.

\bibitem{lefebvre2007diagnosis}
D.~Lefebvre and C.~Delherm.
\newblock Diagnosis of des with {P}etri net models.
\newblock {\em IEEE Transactions on Automation Science and Engineering},
  4(1):114--118, 2007.

\bibitem{lin1994diagnosability}
F.~Lin.
\newblock Diagnosability of discrete event systems and its applications.
\newblock {\em Discrete Event Dynamic Systems}, 4(2):197--212, 1994.

\bibitem{lin2014control}
F.~Lin.
\newblock Control of networked discrete event systems: dealing with
  communication delays and losses.
\newblock {\em SIAM Journal on Control and Optimization}, 52(2):1276--1298,
  2014.

\bibitem{lin2017n}
F.~Lin, W.~Chen, L.~Han, B.~Shen, et~al.
\newblock N-diagnosability for active on-line diagnosis in discrete event
  systems.
\newblock {\em Automatica}, 83:220--225, 2017.

\bibitem{ma2021marking}
Z.~Ma, X.~Yin, and Z.~Li.
\newblock Marking diagnosability verification in labeled {P}etri nets.
\newblock {\em Automatica}, 131:109713, 2021.

\bibitem{nunes2018codiagnosability}
C.~Nunes, M.~Moreira, M.~Alves, L.~Carvalho, and J.~Basilio.
\newblock Codiagnosability of networked discrete event systems subject to
  communication delays and intermittent loss of observation.
\newblock {\em Discrete Event Dynamic Systems}, 28(2):215--246, 2018.

\bibitem{oliveira2022k}
V.~Oliveira, F.~Cabral, and M.~Moreira.
\newblock K-loss robust codiagnosability of discrete-event systems.
\newblock {\em Automatica}, 140:110222, 2022.

\bibitem{pencole2022diagnosability}
Y.~Pencol{\'e} and A.~Subias.
\newblock Diagnosability of event patterns in safe labeled time {P}etri nets: a
  model-checking approach.
\newblock {\em IEEE Transactions on Automation Science and Engineering},
  19(2):1151--1162, 2022.

\bibitem{ran2019enforcement}
N.~Ran, A.~Giua, and C.~Seatzu.
\newblock Enforcement of diagnosability in labeled {P}etri nets via optimal
  sensor selection.
\newblock {\em IEEE Transactions on Automatic Control}, 64(7):2997--3004, 2019.

\bibitem{ran2018codiagnosability}
N.~Ran, H.~Su, A.~Giua, and C.~Seatzu.
\newblock Codiagnosability analysis of bounded {P}etri nets.
\newblock {\em IEEE Transactions on Automatic Control}, 63(4):1192--1199, 2018.

\bibitem{rohloff2005sensor}
K.~Rohloff.
\newblock Sensor failure tolerant supervisory control.
\newblock In {\em Proceedings of the 44th IEEE Conference on Decision and
  Control}, pages 3493--3498. IEEE, 2005.

\bibitem{sampath1995diagnosability}
M.~Sampath, R.~Sengupta, S.~Lafortune, K.~Sinnamohideen, and D.~Teneketzis.
\newblock Diagnosability of discrete-event systems.
\newblock {\em IEEE Transactions on Automatic control}, 40(9):1555--1575, 1995.

\bibitem{sasi2018detectability}
Y.~Sasi and F.~Lin.
\newblock Detectability of networked discrete event systems.
\newblock {\em Discrete Event Dynamic Systems}, 28(3):449--470, 2018.

\bibitem{shu2015supervisor}
S.~Shu and F.~Lin.
\newblock Supervisor synthesis for networked discrete event systems with
  communication delays.
\newblock {\em IEEE Transactions on Automatic Control}, 60(8):2183--2188, 2015.

\bibitem{tai2022new}
R.~Tai, L.~Lin, Y.~Zhu, and R.~Su.
\newblock A new modeling framework for networked discrete-event systems.
\newblock {\em Automatica}, 138:110139, 2022.

\bibitem{takai2021general}
S.~Takai.
\newblock A general framework for diagnosis of discrete event systems subject
  to sensor failures.
\newblock {\em Automatica}, 129:109669, 2021.

\bibitem{takai2017generalized}
S.~Takai and R.~Kumar.
\newblock A generalized framework for inference-based diagnosis of discrete
  event systems capturing both disjunctive and conjunctive decision-making.
\newblock {\em IEEE Transactions on Automatic Control}, 62(6):2778--2793, 2017.

\bibitem{takai2012verification}
S.~Takai and T.~Ushio.
\newblock Verification of codiagnosability for discrete event systems modeled
  by mealy automata with nondeterministic output functions.
\newblock {\em IEEE Transactions on Automatic Control}, 57(3):798--804, 2012.

\bibitem{thorsley2008diagnosability}
D.~Thorsley, T.~Yoo, and H.~Garcia.
\newblock Diagnosability of stochastic discrete-event systems under unreliable
  observations.
\newblock In {\em 2008 American Control Conference}, pages 1158--1165. IEEE,
  2008.

\bibitem{tomola2017robust}
J.~Tomola, F.~Cabral, L.~Carvalho, and M.~Moreira.
\newblock Robust disjunctive-codiagnosability of discrete-event systems against
  permanent loss of observations.
\newblock {\em IEEE Transactions on Automatic Control}, 62(11):5808--5815,
  2017.

\bibitem{tong2021state}
Y.~Tong, J.~Luo, and C.~Seatzu.
\newblock State estimation of discrete-event systems subject to intermittent
  and permanent loss of observations.
\newblock In {\em 2021 60th IEEE Conference on Decision and Control (CDC)},
  pages 1048--1053. IEEE, 2021.

\bibitem{tuxi2022diagnosability}
T.~Tuxi, L.~Carvalho, E.~Nunes, and A.~da~Cunha.
\newblock Diagnosability verification using {LTL} model checking.
\newblock {\em Discrete Event Dynamic Systems}, pages 1--35, 2022.

\bibitem{ushio2016nonblocking}
T.~Ushio and S.~Takai.
\newblock Nonblocking supervisory control of discrete event systems modeled by
  {M}ealy automata with nondeterministic output functions.
\newblock {\em IEEE Transactions on Automatic Control}, 61(3):799--804, 2016.

\bibitem{viana2019codiagnosability}
G.~Viana and J.~Basilio.
\newblock Codiagnosability of discrete event systems revisited: A new necessary
  and sufficient condition and its applications.
\newblock {\em Automatica}, 101:354--364, 2019.

\bibitem{wada2021decentralized}
A.~Wada and S.~Takai.
\newblock Decentralized diagnosis of discrete event systems subject to
  permanent sensor failures.
\newblock {\em Discrete Event Dynamic Systems}, pages 1--35, 2021.

\bibitem{yin2017initial}
X.~Yin.
\newblock Initial-state detectability of stochastic discrete-event systems with
  probabilistic sensor failures.
\newblock {\em Automatica}, 80:127--134, 2017.

\bibitem{yin2017supervisor}
X.~Yin.
\newblock Supervisor synthesis for {M}ealy automata with output functions: {A}
  model transformation approach.
\newblock {\em IEEE Transactions on Automatic Control}, 62(5):2576--2581, 2017.

\bibitem{yin2019robust}
X.~Yin, J.~Chen, Z.~Li, and S.~Li.
\newblock Robust fault diagnosis of stochastic discrete event systems.
\newblock {\em IEEE Transactions on Automatic Control}, 64(10):4237--4244,
  2019.

\bibitem{yin2017decidability}
X.~Yin and S.~Lafortune.
\newblock On the decidability and complexity of diagnosability for labeled
  {P}etri nets.
\newblock {\em IEEE Transactions on Automatic Control}, 62(11):5931--5938,
  2017.

\bibitem{zaytoon2013overview}
J.~Zaytoon and S.~Lafortune.
\newblock Overview of fault diagnosis methods for discrete event systems.
\newblock {\em Annual Reviews in Control}, 37(2):308--320, 2013.

\bibitem{zhou2021detectability}
L.~Zhou, S.~Shu, and F.~Lin.
\newblock Detectability of discrete-event systems under nondeterministic
  observations.
\newblock {\em IEEE Transactions on Automation Science and Engineering},
  18(3):1315--1327, 2021.

\end{thebibliography}
	
\appendix
\emph{Proof of Theorem~\ref{thm:varphi-existence}}
\begin{proof}
	($\Rightarrow$) 
	Suppose that there exists a diagnoser $D$ satisfying conditions C1 and C2, while, for the sake of  contradiction that system $G$ is not $\varphi$-diagnosable. 
	This means that there exists an infinite extended faulty string 
	$s \in \mathcal{L}_{e,F}^{\varphi}(G)$ 
	such that for any prefix $t\in \overline{\{s\}}$, there exists an extended normal string 
	$w \in \overline{\mathcal{L}^{\varphi}_{e}(G)}$ having the same observation with $t$. 
	Since diagnoser $D$ satisfies condition C2, 
	for any $t \in \overline{\{s\}}$, we have $D(\Theta_{\Delta}(t))=0$; otherwise, if $D(\Theta_{\Delta}(t))=1$, we have $D(\Theta_{\Delta}(\omega))=D(\Theta_{\Delta}(t))=1$, which violates condition C2.
	Therefore, for any $t \in \overline{\{s\}}$, we have $D(\Theta_{\Delta}(t))=0$. As a result, condition C1 does not hold for diagnoser $D$, which contradicts the assumption.
	
	($\Leftarrow$) Suppose that system $G$ is $\varphi$-diagnosable. We consider a diagnoser $D:\mathcal{O}(\mathcal{L}(G)) \to \{0,1\}$ defined by: for any $\omega \in \overline{\mathcal{L}^{\varphi}_{e}(G)}$
	\begin{equation} \label{eq:diag}
		\begin{split}
			\!\!\!\!&D(\Theta_{\Delta}(\omega))= \\
			&\!\!\!\!\left\{\!\!
			\begin{array}{ll}
				1  &\text{if  } \forall s \!\in\! \overline{\mathcal{L}^{\varphi}_{e}(G)}:  \Theta_{\Delta}(\omega)\! =\! \Theta_{\Delta}(s) \!\Rightarrow\! \Sigma_{e,F} \!\in\! \Theta_{\Sigma}(s)  \\
				0  &\text{otherwise}.
			\end{array}
			\right. \\  
		\end{split}
	\end{equation}
	We claim that the diagnoser given by equation \eqref{eq:diag} satisfies condition C1 and C2.
	We first show that the diagnoser satisfies condition C2. For any extended normal string $s \in \overline{\mathcal{L}^{\varphi}_{e}(G)}$ and $\Sigma_{e,F} \notin \Theta_{\Sigma}(s)$, by Equation~\eqref{eq:diag}, we have $D(\Theta_{\Delta}(s))=0$, that is, C2 holds.
	To see that C1 holds, we consider an extended faulty string $s \in \mathcal{L}^{\varphi}_{e,F}(G)$. 
	By $\varphi$-diagnosability, there exists a finite prefix $t \in \overline{\{s\}}$ such that any finite extended string $\omega$ having the same observation with $t$ is faulty, i.e., $\Sigma_{e,F} \in \Theta_{\Sigma}(\omega)$. Then, by Equation \eqref{eq:diag}, we have $D(\Theta_{\Delta}(t))=1$, i.e., condition C1 also holds.
\end{proof}

\emph{Proof of Theorem~\ref{thm:verify varphi-diag}}
\begin{proof}
	$(\Leftarrow)$ 
	Assume that there exists a reachable cycle $\pi'= q_V^1 \xrightarrow{\sigma_V^1} q_V^2 \xrightarrow{\sigma_V^2} \cdots \xrightarrow{\sigma_V^{n-1}} q_V^n$ in the verification system $V$ such that    $q_V^{i} \in Q_{m,V}$ and $\theta_1(\sigma_V^j) \neq \varepsilon$ for some $i,j \in \{1,2,\cdots , n\}$, but the system $G$ is $\varphi$-diagnosable. 
	Without loss of generality, let $q_V^1=(q_1,q_2) \in Q_{m,V}$, i.e., $i=1$, and $s_2= \sigma_V^1 \sigma_V^2 \cdots \sigma_V^n$. 
	Since cycle $\pi'$ is reachable, we can find a string $s_1 \in \mathcal{L}^*(V)$ such that $q_V^1 \in f_V(q_{0,V}, s_1)$ where $q_{0,V} \in Q_{0,V}$. 
	After repeating the cycle for infinite number of times, we obtain an infinite string $s=s_1 (s_2)^\omega \in \mathcal{L}^\omega(V)$, over which we get a run $\pi =q_V^0\xrightarrow{\sigma_V^0} \cdots (q_V^1 \xrightarrow{\sigma_V^1} q_V^2 \xrightarrow{\sigma_V^2} \cdots \xrightarrow{\sigma_V^{n-1}} q_V^n )^\omega$.
	Along  run $\pi$, we can extract a path $\rho=q_V^0 \cdots (q_V^1 \cdots q_V^n)^\omega$ where  accepting state $q_V^1 \in Q_{m,V}$ appears for infinite number of times, i.e., $\textsf{Inf}(\rho) \cap Q_{m,V} \neq \emptyset$.
	Since there exists $i \in \{1,2,\cdots , n\}$ such that $\theta_1(\sigma_V^i) \neq \varepsilon$, the first component of sequence $s_l=\theta_1(s)$ is also an infinite extended string. 
	By $q_V^1=(q_1,q_2) \in Q_{m,V}$, we know that $q_1$ is included in $Q_{m,T} \cap Q_{F,T}$ which means  infinite extended string $s_l$ is faulty and accepting in system $T$, that is, $s_l \in \mathcal{L}_{e,F}^\varphi(G)$. 
	On the other hand, we know that $q_2 \in Q_{feas,T}\cap Q_{N,T}$ and by  Equation \eqref{eq:feasible}, the second component of sequence $s_r=\theta_2(s)$ is normal and every prefix of $s_r$ is a prefix of an accepting extended string, i.e., $\forall \omega\in \overline{\{s_r\}}, \omega \in \overline{\mathcal{L}^{\varphi}_{e}(G)}$ and $\Sigma_{e,F}\notin s_r$. 
	Thus, by the first property of the verification system, for any finite prefix $t \in \overline{\{s_l\}}$, there exists $\omega \in \overline{\{s_r\}} \subset \overline{\mathcal{L}^{\varphi}_{e}(G)}$ and $\Sigma_{e,F} \notin \omega$ such that $\Theta_{\Delta}(t)=\Theta_{\Delta}(\omega)$, that is, $G$ is not $\varphi$-diagnosable.

	$(\Rightarrow)$
	Suppose that  system $G$ is not $\varphi$-diagnosable. 
	That is, there exists an infinite faulty extended string $s \in \mathcal{L}_{e,F}^{\varphi}(G)$ such that for any prefix $t \in \overline{\{s\}}$, we have a normal extended string $\omega \in \overline{\mathcal{L}^{\varphi}_{e}(G)}$ with the same observation of $t$, i.e.,
	\begin{align}
		&(\exists s \in \mathcal{L}^{\varphi}_{e,F}(G))(\forall t \in \overline{\{s\}}): \nonumber \\
		&(\exists \omega \in \overline{\mathcal{L}^{\varphi}_{e}(G)} )  
		[\Theta_{\Delta}(w)=\Theta_{\Delta}(t) \wedge \Sigma_{e,F}\notin w]. \vspace{8pt} \label{eq:not diag}
	\end{align}
Thus, for  infinite faulty extended string $s \in \mathcal{L}^\varphi_{e,F}(G)$, there exists a path $\rho$ along   string $s$ in system $T$ such that $\textsf{Inf}(\rho) \cap Q_{m,T} \cap Q_T^F \neq \emptyset$. 
By the accepting condition of B\"{u}chi automata, we can select a state $q_1 \in \textsf{Inf}(\rho) \cap Q_{m,T}\cap Q_T^F$ such that there is a finite faulty prefix $t \in \overline{\{s\}}$, over which the path can reach state $q_1 \in Q_T$ exceeding $n=|\Sigma_e| \times (|\Sigma_e \cup \{\varepsilon\}|) \times |Q_T|$ times.
By Equation \eqref{eq:not diag}, there exists a normal extended string $\omega \in \overline{\mathcal{L}^{\varphi}_{e}(G)}$ having the same observation with $t$.
By Equation \eqref{eq:feasible},  set $\delta_T( Q_{0,T},\omega) \cap Q_{feas,T}\cap Q_{N,T}$ is not empty.
Thus, by the second property of the verification system, there exists a string $s_V \in \mathcal{L}(V)$ in verification system $V$ such that $\theta_1(s_V)=t, \theta_2(s_V)=\omega$ and the run over $s_V$ visits state $q_V \in \{(q_1,q_2): q_2 \in \delta_T( Q_{0,T},\omega) \cap Q_{feas,T}\cap Q_{N,T}\}$ by event $\sigma \in \{(\sigma_V^1, \sigma_V^2) \in \Sigma_V: \sigma_V^1 \neq \varepsilon\}$ exceeding $n$ times, that is, the structure $\xrightarrow{\sigma} q_V$ appears in $\pi$ exceeding $n$ times. 
By definition of   set $Q_{m,V}$, we know that  state $q_V$ is included in $Q_{m,V}$.
Since $|\{(q_1,q_2): q_2 \in \delta_T( Q_{0,T},\omega) \cap Q_{feas,T}\cap Q_{N,T}\}| \leq |Q_T|$ and $|\{(\sigma_V^1, \sigma_V^2) \in \Sigma_V: \sigma_V^1 \neq \varepsilon\}| \leq |\Sigma_e| \times (|\Sigma_e \cup \{\varepsilon\}|)$, by the Pigeonhole Principle, there exits a cycle $\pi = q_V^1 \xrightarrow{\sigma_V^1} q_V^2 \xrightarrow{\sigma_V^2} \cdots \xrightarrow{\sigma_V^{n-1}} q_V^n$ such that there exist $i,j \in \{1,2,\cdots , n-1\}$ satisfying $q_V^{i} \in Q_{m,V}$ and $\theta_1(\sigma_V^j) \neq \varepsilon$. 
\end{proof}
\end{document}